\documentclass{article}
\usepackage[utf8]{inputenc}
\usepackage[english]{babel}
\usepackage[pdfusetitle, colorlinks=true, linktoc=page, allcolors=blue]{hyperref}
\usepackage{inconsolata}
\usepackage{graphicx}
\usepackage{subcaption}
\usepackage{listings}
\usepackage{xcolor}
\usepackage{indentfirst}
\usepackage[numbers]{natbib}
\usepackage{amsmath,amssymb, amsthm}
\usepackage{multirow}
\usepackage{threeparttable}
\usepackage{makecell}
\usepackage{caption}
\usepackage{enumitem}
\usepackage{array}
\usepackage[left=2cm, right=2cm,
top=2cm,bottom=2cm,bindingoffset=0cm]{geometry}
\newtheorem{theorem}{Theorem}

\newtheorem{lemma}{Lemma}

\newtheorem*{corollary*}{Corollary}
\theoremstyle{definition}
\newtheorem{definition}{Definition}
\newtheorem{clproof}{Proof of Controllability Condition}

\renewcommand{\Re}{{\rm Re}}
\renewcommand{\Im}{{\rm Im}} 
\newcommand{\D}{{\rm D}} 
\newcommand{\ad}{{\rm \operatorname{ad}}} 

\begin{document}
\title{A complete classification of operator controllability for closed three-level quantum systems driven by a single control}
\author{
{Sergey Kuznetsov}$^{1}$ and {Alexander Pechen}$^{1}$\\[0.5em]
$^{1}$ Department of Mathematical Methods for Quantum Technologies \\
Steklov Mathematical Institute of Russian Academy of Sciences \\
Moscow, 119333 Russia
}

\date{}

\maketitle

% ABSTRACT
\begin{abstract}
In this work, we study operator controllability of finite-dimensional closed quantum systems driven by a single coherent control. Two general Lie-algebraic results on controllability of such systems are proved. The first gives sufficient conditions under which controllability of the system is preserved when controlled energy-level shifts are added. The second gives a sufficient basis-element condition under which the dynamical Lie algebra contains $\mathfrak{su}(N)$ and the system is therefore operator controllable. These results, together with known controllability conditions, are applied to obtain a complete classification of operator controllability of all three-level quantum systems in terms of their free and interaction Hamiltonian matrix elements. This classification covers all combinations of energy-level and level-spacing degeneracies, two or three allowed transitions, as well as the presence or absence of controlled energy-level shifts, and gives necessary and sufficient conditions for operator controllability in each case.
\end{abstract} 

% SECTION 1
\section{Introduction}
Quantum control concerns the manipulation of quantum systems to achieve prescribed dynamical objectives. It has broad applications in quantum technologies, including quantum computing and information processing, atomic and molecular physics and chemistry, spectroscopy and magnetic resonance, quantum sensing, and metrology. The field has attracted considerable attention because it addresses fundamental limits on the control of quantum dynamics and is central to the continuing development of quantum technologies~\cite{CTLSCC_bib:RiceZhao2000, CTLSCC_bib:Tannor2007, CTLSCC_bib:BrifChakrabartiRabitz2010, CTLSCC_bib:ShapiroBrumer2012, CTLSCC_bib:AlessandroBook, CTLSCC_bib:KochrRewiev2022, CTLSCC_bib:Kuprov2023}.

Controllability is a central question in quantum control for both closed and open quantum systems. It asks whether the available controls can steer a quantum system from a given initial state to a desired target state or, more generally, between prescribed sets of states. It is fundamental both theoretically and practically, since controllability must be established before one attempts to synthesize controls. If a desired target is unreachable, then the control objective, the available controls, or the system design must be modified to overcome the limitation. Controllability is therefore a key prerequisite for the analysis and design of controlled quantum dynamics~\cite{CTLSCC_bib:AlessandroBook, CTLSCC_bib:DirrHelmke2008, CTLSCC_bib:DongPetersen2010, Altafini_Ticozzi_2012}.

Controllability has been investigated in various settings for closed and open systems. Particular attention has been devoted to the investigation of operator, or complete, controllability of finite-dimensional closed quantum systems. Here, one considers an $ N $--level system without dissipation evolving under coherent control. The system is called operator controllable when the available interaction can generate any unitary evolution up to a physically irrelevant global phase. This setting is especially useful because the problem reduces to controllability on the compact Lie group $ \mathrm{U}(N) $. The Lie Algebra Rank Condition (LARC)~\cite{CTLSCC_bib:SussmannJurdjevic1972} can then be applied and, because the Lie group $ \mathrm{U}(N) $ is compact, provides a necessary and sufficient condition for controllability. The geometric and systems-theoretic foundations of this approach were developed in early work on quantum-mechanical controllability \cite{CTLSCC_bib:Huang1983, CTLSCC_bib:Ramakrishna1995}, while the relations between different notions of controllability for multilevel quantum systems were clarified in \cite{CTLSCC_bib:AlbertiniAlessandro2003}. Constructive methods for obtaining controls once Lie-algebraic controllability has been established were considered, for example, in \cite{CTLSCC_bib:DAlessandro2009}. General expositions of modeling, controllability, and control synthesis for closed and open quantum systems can be found in \cite{CTLSCC_bib:AlessandroBook, Altafini_Ticozzi_2012}.

The LARC is a complete criterion, but it is often too implicit for structural analysis. A major goal is therefore to obtain more explicit controllability conditions in terms of the Hamiltonian matrix elements. Such criteria were derived for several classes of systems: in particular, sufficient conditions for finite-dimensional quantum systems driven by a single control field with dipole-type interaction were obtained \cite{CTLSCC_bib:SchirmerSolomonFu2001, Fu_Schirmer_Solomon_2001}; criteria for three-level systems with two allowed transitions between levels (i.e., when the interaction Hamiltonian has two nonzero independent transition matrix elements) and without controlled energy-level shifts (i.e., when the interaction Hamiltonian has only off-diagonal matrix elements) were described~\cite{CTLSCC_bib:KuznetsovPechen2023}; some four-level systems and corresponding controllability conditions were also considered \cite{CTLSCC_bib:MyachkovaPechen2023, CTLSCC_bib:KuznetsovPechen2022}. General conditions formulated in terms of energy-level structure and allowed transitions were established \cite{CTLSCC_bib:TuriniciRabitz2001, CTLSCC_bib:Altafini2002, CTLSCC_bib:TuriniciRabitz2003}. The Lie-algebraic symmetry-based approach to the classification of uncontrollable systems was developed~\cite{CTLSCC_bib:PolackSuchowskiTannor2009}. The role of degeneracies was studied, in particular, for systems with a twofold energy-level degeneracy \cite{Zhang_Fu_2010}. Symmetry principles, dynamical Lie algebras (i.e., Lie algebras generated by the free and interaction Hamiltonians of the controlled system), and operator controllability were studied~\cite{Zeier_Schulte-Herbruggen_2011}. The identification of dynamical Lie algebras for finite-dimensional closed quantum systems, including cases where the generated algebra is a proper subalgebra of $\mathfrak{su}(N)$, was addressed \cite{Schirmer_Pullen_Solomon_2002}. Closely related general controllability questions for invariant systems on Lie groups were considered~\cite{Sachkov_2000}. Controllability of open quantum systems has also been studied; for an analysis of an open two-level system, see~\cite{CTLSCC_bib:LokutsievskiyPechen1} and the references therein.

Although general controllability criteria do exist and controllability of some classes of three-level quantum systems has been investigated, a complete classification of operator controllability for three-level systems in terms of the matrix elements of their free and interaction Hamiltonians appears to be missing. In this work, we fill this gap by providing this classification. For this, we prove two new general Lie-algebraic controllability results. The first gives sufficient conditions for a controllable system to preserve its controllability after adding controlled energy-level shifts --- in general, a nontrivial feature, since adding control over the energy levels to a controllable system in some cases can break its controllability. The second result gives a convenient sufficient basis-element condition under which the dynamical Lie algebra contains $\mathfrak{su}(N)$ and the system is therefore operator controllable. Combining these results with known controllability criteria, we provide a complete classification of operator controllability for three-level systems, covering all combinations of energy-level degeneracies, level-spacing degeneracies, two or three allowed transitions, as well as the presence or absence of controlled energy-level shifts, and providing necessary and sufficient conditions for operator controllability for each case.

The article is organized as follows. Section~\ref{CTLSCC_sec:controllabilitynotions} introduces the notion of operator controllability and recalls the known Lie-algebraic controllability conditions for finite-dimensional closed quantum systems that are required for our analysis. In Section~\ref{CTLSCC_sec:mygeneralresults}, we establish two new results for finite-level quantum systems: a condition ensuring the preservation of controllability when controlled energy-level shifts are added to a controllable system (Subsection~\ref{CTLSCC_subsec:shiftsstatement}) and a sufficient basis-element condition under which the dynamical Lie algebra contains $\mathfrak{su}(N)$  (Subsection~\ref{CTLSCC_subsec:algebrastatement}). Section~\ref{CTLSCC_sec:threelevelsystems} provides a complete operator controllability classification of closed three-level systems driven by a single coherent control. In particular, Subsection~\ref{CTLSCC_subsec:systemtypes} describes general classes of controlled closed three-level quantum systems. Subsections~\ref{CTLSCC_subsec:twotransitions} and~\ref{CTLSCC_subsec:threetransitions} present controllability conditions for systems with two and three allowed transitions, respectively. Subsection~\ref{CTLSCC_subsec:highlight} illustrates the role of all conditions in the theorem on the preservation of controllability under the addition of controlled energy-level shifts. Subsection~\ref{CTLSCC_subsec:finalclassification} summarizes the resulting controllability classification of three-level systems. Finally, Section~\ref{CTLSCC_sec:conclusion} concludes the article.

% SECTION 2
\section{Controllability notion and known Lie-algebraic criteria}\label{CTLSCC_sec:controllabilitynotions}

In this section, we formulate the operator controllability problem for finite-dimensional closed quantum systems driven by a single coherent control and recall the Lie algebra rank condition for the controllability of such systems. Then we introduce the Hamiltonian matrix-element notation, the skew-Hermitian basis and its commutation relations, and the graph-based structural properties used throughout the paper. Finally, we summarize the known Lie-algebraic controllability conditions which are used for the subsequent classification.

The dynamics of an $ N $-dimensional closed quantum system driven by a single coherent control is governed by the Schr\"odinger equation for the unitary \text{evolution operator} $ U_{f}(t) \in \mathbf{U}(N) $
\begin{equation}\label{CTLSCC_eq:dynamicssingle}
        i \frac{dU_{f}(t)}{dt}  = \left(H + f(t) \, V\right) U_{f}(t), \quad U_f(0)=\mathbb{I},
\end{equation}
where $ H = H^{\dagger} $ and $ V = V^{\dagger} $ are the \textit{free} (or \textit{drift}) and  the \textit{interaction} (or \textit{control}) parts of the system's Hamiltonian, respectively, and $ f(t) \in \mathcal{U}([0,T], \mathbb{R}) $ is a \textit{control function}. We assume that the free Hamiltonian $ H $ is not proportional to the identity matrix $ \mathbb{I} $, because such a free Hamiltonian would produce trivial quantum control dynamics. The eigenvalues of $H$, denoted by $ E_1, \, \dots, \, E_N $, are (possibly degenerate) energy levels of the system. Off-diagonal matrix elements of the interaction Hamiltonian $ V $ in the eigenbasis of $ H $, denoted by $ v_{ij} $, where $ i $ and $ j $ label the eigenvectors of $ H $, are transition matrix elements. Diagonal matrix elements of the interaction Hamiltonian $ V $ in the eigenbasis of $H$, $v_{kk}=:v_{k} $, are controlled energy-level shifts. For the class of admissible controls $ \mathcal{U} $, we consider any function class containing piecewise-constant controls. The results reviewed in this section are of a general nature and do not depend on a particular choice of such $ \mathcal{U} $. Throughout this article, we identify a controlled quantum system with the pair $ (H, V) $.

There are various notions of controllability for finite-dimensional closed quantum systems \cite{CTLSCC_bib:AlbertiniAlessandro2003}. In this work, we focus on the notion of (complete) operator controllability which is important for practical applications --- a system is controllable if, using admissible controls, one can synthesize any unitary evolution (up to a global phase) in a finite time.
\begin{definition}\label{CTLSCC_def:completecontrollability}
		A closed quantum system $ (H, \, V) $ is called \textit{operator} (or \textit{completely}) \textit{controllable} if, for any target unitary operator $ U \in \mathbf{U}(N) $, there exists a final time $ T > 0 $ and a control $ f(t) \in \mathcal{U}([0,T], \mathbb{R}) $ such that the solution $ U_{f} (T) $ of the equation \eqref{CTLSCC_eq:dynamicssingle} satisfies
		\begin{equation}
			U = e^{i \varphi} U_{f} (T),
		\end{equation}
		with some (physically irrelevant) global phase $ \varphi \in \mathbb R $.
\end{definition}
The terms controllability and controllable system are used in this sense throughout the paper.

For this controllability notion, a well-known general geometric controllability criterion is formulated in terms of the quantum system's dynamical Lie algebra.
\begin{definition}
    The \textit{dynamical Lie algebra} $ \mathcal{L}(iH, iV) $ of the quantum system $ (H, V) $ is the Lie algebra generated by~$ iH $ and $ iV $ \cite{CTLSCC_bib:Ramakrishna1995}.
\end{definition}
\begin{theorem}[LARC]\label{CTLSCC_th:maincriterion}
    Let the free and interaction Hamiltonians be traceless, $ \operatorname{Tr} H = \operatorname{Tr} V = 0 $. Then the system $ (H, \, V) $ is controllable iff 
    \begin{equation}
        \mathcal{L}(iH, \, iV) = \mathfrak{su}(N).
    \end{equation}
    Otherwise, it is controllable iff
    \begin{equation}
        \mathcal{L}(iH, \, iV) = \mathfrak{u}(N).
    \end{equation}
\end{theorem}
Since the Lie algebra $ \mathfrak{u}(N) $ has the decomposition $ \mathfrak{u}(N) = \mathfrak{su}(N) \oplus i \mathbb{R} \mathbb{I} $, the theorem states that the system is controllable iff
\begin{equation}
    \mathcal{L}(iH, \, iV) \supset \mathfrak{su}(N).
\end{equation}
In other words, since the global phase of the evolution operator is physically irrelevant, operator controllability is completely determined by whether the dynamical Lie algebra $ \mathcal{L}(iH, \, iV) $ contains $ \mathfrak{su}(N) $ or not.

Theorem~\ref{CTLSCC_th:maincriterion} can be reformulated in the following form, which is convenient for numerical controllability verification.
\begin{corollary*}
    The system $ (H, V) $ is controllable iff $ \dim\mathcal{L}(iH, \, iV) \geqslant N^{2} - 1 $.  
\end{corollary*}

The LARC provides a straightforward way to verify whether a given quantum system is controllable by performing a computational test to determine the number of linearly independent elements in the dynamical Lie algebra. Details on the algorithm can be found in \cite{CTLSCC_bib:AlessandroBook, CTLSCC_bib:SchirmerSolomonFu2001}. However, this approach has several significant drawbacks:

\begin{enumerate}
    \item The numerical test does not allow for the formulation of general analytical conditions for controllability in terms of the energy levels and structure of the transitions of the system under control. Such conditions would be highly useful for various applications (e.g., for studies of quantum control landscapes, such as that performed in \cite{CTLSCC_bib:VolkovMyachkovaPechen2025}).
    \item The numerical test does not allow one to determine general patterns of controllability. In particular, it does not show how the structural features of the free and interaction Hamiltonians affect the presence or absence of controllability.
    \item The numerical test does not allow one to detect “weakly” controllable systems, i.e., controllable systems that are structurally close to uncontrollable systems, since determining controllability by testing linear independence is an inherently ill-conditioned numerical problem.
\end{enumerate}

Therefore, an important task is to derive controllability conditions in a form that does not have these limitations. One possible approach is to formulate controllability conditions directly in terms of the matrix elements of the system Hamiltonian. 

Let 
\begin{equation}
    \mathcal{B} = \left\{B_{ij}^{\Re}, B_{ij}^{\Im} \right\}_{i < j} \cup \left\{B_{k}^{\D} \right\}_{k}
\end{equation}
be the set of matrices whose elements are defined as
\begin{equation}\label{CTLSCC_eq:basisB}
    \begin{split}
        &B_{ij}^{\Re} = \left(\mu_{ji} - \mu_{ij}\right), \\
        &B_{ij}^{\Im} = i \left( \mu_{ij} + \mu_{ji} \right), \\
        &B_{k}^{\D} = i \mu_{kk},
    \end{split}
\end{equation}
where $ \mu_{ij} $ is an $ N \times N $ matrix with $ \mu_{ij}^{(\alpha \beta)} = \delta_{i \alpha} \delta_{j \beta} $. From this definition, it directly follows that $ B_{ij}^{\Re} = -B_{ji}^{\Re} $, $ B_{ij}^{\Im} = B_{ji}^{\Im }$, $ B_{kk}^{\Re} = 0 $, and $ B_{kk}^{\Im} = 2 B_{k}^{\D} $. These matrices satisfy the following commutation rules 
\begin{align}
    &\left[B_{i_{1} j_{1}}^{\Re}, \, B_{i_{2} j_{2}}^{\Re} \right] = \delta_{j_{1} j_{2}} B_{i_{1} i_{2}}^{\Re} - \delta_{j_{1} i_{2}} B_{i_{1} j_{2}}^{\Re} - \delta_{i_{1} j_{2}} B_{j_{1} i_{2}}^{\Re} + \delta_{i_{1} i_{2}} B_{j_{1} j_{2}}^{\Re}, \label{CTLSCC_eq:ReRecomm}  \\ 
    &\left[B_{i_{1}j_{1}}^{\Re}, \, B_{i_{2}j_{2}}^{\Im} \right] =  -\delta_{j_{1} j_{2}} B_{i_{1} i_{2}}^{\Im} - \delta_{j_{1} i_{2}} B_{i_{1} j_{2}}^{\Im} + \delta_{i_{1} j_{2}} B_{j_{1} i_{2}}^{\Im} + \delta_{i_{1} i_{2}} B_{j_{1} j_{2}}^{\Im}, \label{CTLSCC_eq:ReImcomm}  \\ 
    &\left[B_{i_{1}j_{1}}^{\Im}, \, B_{i_{2}j_{2}}^{\Im} \right] = \delta_{j_{1} j_{2}} B_{i_{1} i_{2}}^{\Re} + \delta_{j_{1} i_{2}} B_{i_{1} j_{2}}^{\Re} + \delta_{i_{1} j_{2}} B_{j_{1} i_{2}}^{\Re} + \delta_{i_{1} i_{2}} B_{j_{1} j_{2}}^{\Re}, \label{CTLSCC_eq:ImImcomm}  \\
    &\left[B_{ij}^{\Re}, \, B_{k}^{\D} \right] = (\delta_{ik}-\delta_{jk})B_{ij}^{\Im}, \label{CTLSCC_eq:ReDcomm}\\
    &\left[B_{ij}^{\Im}, \, B_{k}^{\D} \right] = (\delta_{jk}-\delta_{ik})B_{ij}^{\Re}, \label{CTLSCC_eq:ImDcomm} \\
    &\left[B_{k_{1}}^{\D}, \, B_{k_{2}}^{\D} \right] = 0. \label{CTLSCC_eq:DDcomm} 
\end{align}
To shorten the notation, we denote an unspecified element of the pair $ \{ B_{ij}^{\Re}, B_{ij}^{\Im} \} $ as $ B_{ij} $. The inclusion $ B_{ij} \in \mathcal{L} $ for an arbitrary set $ \mathcal{L} $ is understood as $ \exists B_{ij}: B_{ij} \in \{ B_{ij}^{\Re}, B_{ij}^{\Im} \} $ such that $ B_{ij} \in \mathcal{L} $.

The set of matrices $ \mathcal{B} $ forms a basis of the Lie algebra $ \mathfrak{u}(N) $ and can be used to decompose the skew-Hermitian matrices $ iH $ and $ iV $, corresponding to the free and interaction Hamiltonians, respectively, in the chosen eigenbasis of $ H $. These decompositions take the form
\begin{align}
    &i H = i \begin{pmatrix} 
        E_{1} & 0 & \ldots & 0 \\
        0 & E_{2} & \ldots & 0 \\
        \vdots & \vdots & \ddots & \vdots \\
        0 & 0 & \ldots & E_{N} 
    \end{pmatrix} = \sum\limits_{k} E_{k} B_{k}^{\D}, \label{CTLSCC_eq:hdecomposition} \\[0.2cm]
    &i V = i \begin{pmatrix}
        v_{1} & v_{12} & \ldots & v_{1N} \\
        v_{12}^{*} & v_{2} & \ldots & v_{2N} \\
        \vdots & \vdots & \ddots & \vdots \\
        v_{1N}^{*} & v_{2N}^{*}& \ldots & v_{N}
    \end{pmatrix} = \sum\limits_{i < j} \left(v_{ij}^{\Re} B_{ij}^{\Im}  + v_{ij}^{\Im} B_{ij}^{\Re}\right) + \sum\limits_{k} v_{k} B_{k}^{\D}, \label{CTLSCC_eq:vdecomposition} 
\end{align}
where $ v_{ij}^{\Re} := \Re(v_{ij})$ and $ v_{ij}^{\Im} := \Im(v_{ij}) $. A common convention for the energy-level eigenvalues of the free Hamiltonian $H$ is to arrange them in nondecreasing order $ E_{1} \leqslant E_{2} \leqslant \ldots \leqslant E_{N} $. However, in this work, we do not use this convention and present all the formulations and results without assuming this order; such an approach simplifies the construction of the complete operator controllability classification. For convenience, when $ j > i $, $ v_{ji} $ denotes the complex conjugate of the upper-triangular element $ v_{ij} $
\begin{equation}
    v_{ji} = v_{ij}^{*}.
\end{equation}

With these notations, the operator controllability problem for a closed $ N $-level system reduces to deriving controllability conditions in terms of the energy levels $ E_{k} $, the transition matrix elements $ v_{ij} $, and the controlled energy-level shifts $ v_{k} $. These quantities have the standard physical meaning:
\begin{enumerate}
    \item The energy levels $ E_{k} $ are the diagonal matrix elements of the free Hamiltonian $ H $. However, the physical characteristics of the system are not the values of the energy levels but their differences $ \Delta E_{ij}~=~E_{i} - E_{j} $, whose absolute values divided by the Planck constant are \textit{Bohr} or \textit{transition frequencies}. These transition frequencies, rather than the eigenvalues themselves, affect the controllability properties. This is manifested in the fact that, due to the basic properties of the Lie bracket operation, a uniform shift of all energy levels (corresponding to an additive transformation of the Hamiltonian $ H' = H + c \mathbb{I} $ by the identity matrix $ \mathbb{I} $) as well as any uniform scaling of the spectrum (i.e., multiplication by a nonzero scalar factor $ H\to H'=cH $) either does not change the dynamical Lie algebra or modifies it by adding the identity component $ \mathbb I $, which does not change the controllability property of the system.
    \item Off-diagonal (transition) matrix elements $ v_{ij} $ determine the strength of the  \textit{direct population transfer} between energy levels $ i $ and $ j $ under a coherent control.
    \item Diagonal matrix elements $ v_{k} $ describe \textit{shifts of the original energy levels} induced by a coherent control (we refer to them as controlled energy-level shifts). They can also affect controllability, since they modify the initial energy-level structure determined by the energy levels $ E_{k} $. 
\end{enumerate}

Taking into account Theorem~\ref{CTLSCC_th:maincriterion} and the physical interpretation of the free Hamiltonian $ H $, we consider the nonzero-trace and traceless cases within a unified framework, treating $ \mathcal{L}(i H, i V) $ as its extension $ \mathcal{L}(i H, i V) + \operatorname{span}(i \mathbb{I}) $, since these cases are equivalent from the standpoint of controllability. In other words, we assume by convention that the element $ i\mathbb{I} $ belongs to the dynamical Lie algebra and that controllability corresponds to the case when this algebra is the Lie algebra $ \mathfrak{u}(N) $. This provides a more concise exposition of the proofs of our results.

The structure of the energy-level distribution, control-induced population transfer, and controlled energy-level shifts can be conveniently described using the structural graph $ \mathcal{G}_{V} $
\begin{equation}
    V(\mathcal{G}_{V}) = \left\{1, \, 2, \, \ldots, \, N\right\}, \quad E(\mathcal{G}_{V}) = \left\{i \leftrightarrow j: \, v_{ij} \neq 0 \right\} \, \cup \, \left\{ k \leftrightarrow k: \, v_{k} \neq 0 \right\},
\end{equation}
where $ V(\cdot) $ is the set of vertices of the graph and $ E(\cdot) $ is the set of edges (the notation $ i \leftrightarrow j $ refers to an edge between the vertices $ i $ and $ j $).   

We will use the energy levels $ E_{k} $, the transition matrix elements $ v_{ij} $, the controlled energy-level shifts $ v_{k} $, and the structural graph $ \mathcal{G}_{V} $ to define the structural properties of the system that are relevant to controllability.
\begin{definition}\label{CTLSCC_def:structuralfeatures1}
    Consider a system $ (H, V) $ with the free Hamiltonian $ H $, the interaction Hamiltonian $ V $, and their matrix representations \eqref{CTLSCC_eq:hdecomposition}, \eqref{CTLSCC_eq:vdecomposition}. The system is said to have an \textit{energy-level degeneracy} if
    \begin{equation}
        \exists i \neq j: \ E_{i} = E_{j}.
    \end{equation}
    Following the algebraic terminology~\cite{CTLSCC_bib:Altafini2002}, a system with no energy-level degeneracies is called \textit{regular}. The system is said to have a \textit{level-spacing degeneracy} if it has equal Bohr frequencies
    \begin{equation}
        \exists i_{1}, j_{1}, i_{2}, j_{2}: \quad |E_{i_{1}} - E_{j_{1}}| = |E_{i_{2}} - E_{j_{2}}|,
    \end{equation}
    where $ i_{1} \neq j_{1}, \ i_{2} \neq j_{2}, \, (i_{1}, \, j_{1}) \neq  (i_{2}, \, j_{2}), \, (j_{2}, \, i_{2}) $. A system with no level-spacing degeneracies is called \textit{strongly regular}.
\end{definition}
If the system is strongly regular, it is regular. Indeed, if $ E_{i} = E_{j} $ for $ i \neq j $, then $ | E_{k} - E_{j} | = | E_{k} - E_{i} | $ for any $ k \neq i, \, j $.

Theorems~\ref{CTLSCC_th:graphtheorem} and~\ref{CTLSCC_th:graphtheoremstrengthened}, established in~\cite{CTLSCC_bib:Altafini2002}, describe several properties of the structural graph $ \mathcal{G}_{V} $ that are important for controllability.
\begin{theorem}\label{CTLSCC_th:graphtheorem}
    If the system $ (H, V) $ is controllable, then its structural graph $ \mathcal{G}_{V} $ is connected. If the structural graph $ \mathcal{G}_{V} $ is connected and the system is strongly regular, then the system is controllable.  
\end{theorem}
If the system is nonregular, $ H $ has at least one degenerate eigenvalue whose eigenspace has dimension at least two. Consequently, the choice of the eigenbasis of $ H $ is not unique, so the connectivity of the graph $ \mathcal{G}_{V} $ may depend on this choice. A stronger formulation of the necessary condition in Theorem~\ref{CTLSCC_th:graphtheorem} therefore requires the graph $ \mathcal{G}_{V} $ to be connected for every choice of eigenbasis of $ H $~\cite{CTLSCC_bib:PolackSuchowskiTannor2009}. In constructing the classification of three-level systems below, we require the graph  $ \mathcal{G}_{V} $ to be connected only in the eigenbasis under consideration, since any remaining dependence on the choice of eigenbasis is accounted for by the corresponding controllability conditions.

The sufficient controllability condition in Theorem~\ref{CTLSCC_th:graphtheorem} can be strengthened by restricting the absence of level-spacing degeneracies only to those pairs of levels $ (i, j) $ for which the corresponding edge $  i \leftrightarrow j $ belongs to $ E(\mathcal{G}_{V}) $. To do this, we first introduce the strengthened versions of the notions in  Definition \ref{CTLSCC_def:structuralfeatures1}. 

\begin{definition}\label{CTLSCC_def:structuralfeatures2}
    Consider a system $ (H, V) $ with the free Hamiltonian $ H $, the interaction Hamiltonian $ V $, and their matrix representations \eqref{CTLSCC_eq:hdecomposition}, \eqref{CTLSCC_eq:vdecomposition}. The system is said to have an \textit{energy-level degenerate transition} if
    \begin{equation}
        \exists i \neq j: \ E_{i} = E_{j}, \ v_{ij} \neq 0.
    \end{equation}
    A system without transitions between degenerate levels is called \textit{$ V $--regular}. The system is said to have \textit{level-spacing degenerate transitions} if
    \begin{equation}
        \exists i_{1}, j_{1}, i_{2}, j_{2}: \quad |E_{i_{1}} - E_{j_{1}}| = |E_{i_{2}} - E_{j_{2}}|, \quad v_{i_{1} j_{1}}, \, v_{i_{2} j_{2}} \neq 0, 
    \end{equation}
    where $ i_{1} \neq j_{1}, \ i_{2} \neq j_{2}, \, (i_{1}, \, j_{1}) \neq  (i_{2}, \, j_{2}), \, (j_{2}, \, i_{2}) $. A $ V $--regular system with no level-spacing degenerate transitions is called \textit{$ V $--strongly regular}.
\end{definition}

The qualifier "$ V $--" in terms of Definition~\ref{CTLSCC_def:structuralfeatures2} indicates that regularity and strong regularity are verified only for transitions allowed by $ V $. Unlike ordinary regularity, which follows from the absence of level-spacing degeneracies, $ V $-regularity does not follow from the absence of level-spacing degenerate transitions. Indeed, the latter condition rules out only pairs of allowed transitions with equal Bohr frequencies and does not exclude a single allowed transition between degenerate energy levels. Therefore, $ V $-regularity is imposed explicitly in the definition of $ V $-strong regularity.

The $ V $--strong regularity property can replace strong regularity in the sufficient condition of Theorem \ref{CTLSCC_th:graphtheorem}.
\begin{theorem}\label{CTLSCC_th:graphtheoremstrengthened}
    If the structural graph $ \mathcal{G}_{V} $ is connected and the system is $V$--strongly regular, then the system is controllable.
\end{theorem}

% SECTION 3
\section{Some Lie-algebraic conditions for controllability of general closed quantum systems}\label{CTLSCC_sec:mygeneralresults} 
In this section, we prove two new general Lie-algebraic controllability theorems used later for the three-level classification. The first theorem gives sufficient conditions under which adding controlled energy-level shifts to a controllable system preserves controllability. The second theorem provides a condition, formulated in terms of a subset of the basis elements $ \mathcal{B} $, under which the dynamical Lie algebra $ \mathcal{L}(iH, \, iV) $ contains $\mathfrak{su}(N)$ and the system is therefore operator controllable.

For the analysis, we use a standard linear-algebra lemma to show that a given subset of $ \mathcal{B} $ is contained in $ \mathcal{L}(iH,iV) $.

\begin{lemma}\label{CTLSCC_th:lemmainvertiblebasis}
Let $ L_{1}, \, L_{2}, \, \ldots, L_{M} \in \mathcal{L}(iH,iV) $ be elements of some dynamical Lie algebra which are expressed as linear combinations of some matrices $ B_{1}, \, B_{2}, \, \ldots, \, B_{M} $ as
\begin{equation}
    L_{m} = \sum\limits_{m'=1}^M s_{mm'} B_{m'},
\end{equation}
where $ s_{mm'} \in \mathbb{R} $. If the matrix $ S = (s_{mm'})_{m, m' = 1}^{M} $ is invertible,
\begin{equation}\label{CTLSCC_eq:LSdettechniqueending}
    \det S \neq 0,
\end{equation}
then $ B_{1}, \, B_{2}, \, \ldots, \, B_{M} \in \mathcal{L}(iH,iV) $.
\end{lemma}

This lemma is used in the proof of Theorem~\ref{CTLSCC_th:shifts} and in the proofs of Controllability Conditions~\ref{CTLSCC_th:twonotregular1},~\ref{CTLSCC_th:threenotstrong}, and~\ref{CTLSCC_th:threenotregular} in Theorems~\ref{CTLSCC_th:two_transitions_classification},~\ref{CTLSCC_th:three_transitions_classification}.

\subsection{Preservation of controllability under added controlled energy-level shifts} \label{CTLSCC_subsec:shiftsstatement}

The first theorem provides a sufficient condition for controllability after the addition of controlled energy-level shifts.
\begin{theorem}\label{CTLSCC_th:shifts}
    Consider a $ V $--regular system $ (H, V) $ and a system $ (H, V') $  with the same free Hamiltonian $ H $. Let the interaction Hamiltonians $ V $ and $ V' $ of these systems be such that  
    \begin{equation}\label{CTLSCC_eq:technoname8}
            \begin{split}
                & iV = \sum\limits_{i < j} \left( v_{ij}^{\Im} B_{ij}^{\Re} + v_{ij}^{\Re} B_{ij}^{\Im} \right) + \sum\limits_{k \in K} v_{k} B_{k}^{\D}, \\
                & iV' = iV + \sum\limits_{k \in K'} v'_{k} B_{k}^{\D}, \quad K \cap K' = \varnothing.
            \end{split}
        \end{equation}
        If the system $ (H, V) $ is controllable and $ \displaystyle \sum\limits_{k \in K} v_{k} B_{k}^{\D} \in \mathcal{L}(iH, \, iV') $, then the system $ (H, V') $ is also controllable.
\end{theorem}
\begin{proof}
    Consider a partition of the set of all allowed transitions (i.e., those determined by nonzero transition matrix elements $ v_{ij} \neq 0 $) of the system into $ P $ classes $ \mathcal{E}_{p} $ (where $p=1,\dots, P$) based on their transition frequencies $ |\Delta E_{p}| $
    \begin{equation}\label{CTLSCC_eq:technoname6}
\mathcal{E}_{p} = \left\{(i, \, j): \ i < j, \, v_{ij} \neq 0, \, |E_{i} - E_{j}| =: |\Delta E_{p}| \right\}, \end{equation}
where the corresponding frequencies are pairwise distinct
\begin{equation}
    |\Delta E_{p}| \neq |\Delta E_{p'}| \ \text{ for } \ p \neq p'.
\end{equation}
In terms of the structural graph $ \mathcal{G}_V $, this amounts to partitioning its edges into classes according to their weights, which are determined by the values $ |\Delta E_{p}| $.
    
Then, the matrix $ iV' $ can be written as
    \begin{equation}
        i V' = \sum\limits_{p} \tilde{B}_{p} + \sum\limits_{k \in K} v_{k} B_{k}^{\D} + \sum\limits_{k \in K'} v'_{k} B_{k}^{\D},
    \end{equation}
    where
    \begin{equation}
        \tilde{B}_{p} :=  \sum\limits_{(i, \, j) \in \mathcal{E}_{p}} \left( v_{i j}^{\Im} B_{i j}^{\Re} + v_{i j}^{\Re} B_{i j}^{\Im} \right).
    \end{equation}
    According to the commutation relations \eqref{CTLSCC_eq:ReDcomm} and \eqref{CTLSCC_eq:ImDcomm},
    \begin{equation}
        \ad_{i H}^{2}(B_{ij}) = -(E_{i} - E_{j})^{2} B_{ij} 
    \end{equation}
    and consequently, 
    \begin{equation}
        \ad_{i H}^{4m}(\tilde{B}_{p}) = \Delta E_{p}^{4m}  \tilde{B}_{p}.
    \end{equation}
    This allows us to consider column vectors of matrices $ B = \begin{pmatrix}
        \tilde{B}_{1} & \tilde{B}_{2} & \ldots & \tilde{B}_{P} 
    \end{pmatrix}^{T} $ and $ L = \begin{pmatrix}
        L_{1} & L_{2} & \ldots & L_{P} 
    \end{pmatrix}^{T} $, where
    \begin{equation}
        L_{m} = \ad_{iH}^{4m} \left( i V' \right) = \sum\limits_{p} \ad_{iH}^{4m} \big(\tilde{B}_{p} \big) + \sum\limits_{k \in K} v_{k} \, \underbrace{\ad_{iH}^{4m} \big(B_{k}^{\D}\big)}_{0} +  \sum\limits_{k \in K'} v'_{k} \, \underbrace{\ad_{iH}^{4m} \big(B_{k}^{\D}\big)}_{0} = \sum\limits_{p} \Delta E_{p}^{4m}  \tilde{B}_{p}
    \end{equation}
    for $ m = 1, 2, \, \ldots, \, P $. Therefore, the matrix $ S $ in Lemma~\ref{CTLSCC_th:lemmainvertiblebasis} has the form 
    \begin{equation}\label{CTLSCC_eq:technoname5}
        S = \begin{pmatrix}
             \Delta E_{1}^{4} & \Delta E_{2}^{4} & \ldots & \Delta E_{P}^{4} \\
             \Delta E_{1}^{8} & \Delta E_{2}^{8} & \ldots & \Delta E_{P}^{8} \\
             \vdots & \vdots & \ddots & \vdots \\
             \Delta E_{1}^{4P} & \Delta E_{2}^{4P} & \ldots  & \Delta E_{P}^{4P}
        \end{pmatrix}.
    \end{equation}
    The Vandermonde determinant formula gives for the determinant of the matrix \eqref{CTLSCC_eq:technoname5} the expression
    \begin{equation}
        \det S = \prod_{p} \Delta E_{p}^{4} \ \det \begin{pmatrix}
             1 & 1 & \ldots & 1 \\
             \Delta E_{1}^{4} & \Delta E_{2}^{4} & \ldots & \Delta E_{P}^{4} \\
             \vdots & \vdots & \ddots & \vdots \\
             \Delta E_{1}^{4(P-1)} & \Delta E_{2}^{4(P-1)} & \ldots  & \Delta E_{P}^{4(P-1)}
        \end{pmatrix} = \prod_{p} \Delta E_{p}^{4} \ \prod_{p < p'} \left( \Delta E_{p'}^{4} -  \Delta E_{p}^{4} \right).
    \end{equation}
    $ V$--regularity implies that $ \displaystyle \prod_{p} \Delta E_{p}^{4} \neq 0 $, while $ \displaystyle \prod_{p < p'} \left( \Delta E_{p'}^{4} -  \Delta E_{p}^{4} \right) \neq 0 $ because the classes $ \mathcal{E}_{p} $ in the partition defined by the formula \eqref{CTLSCC_eq:technoname6} are associated with pairwise distinct absolute transition frequencies. Thus, $ \det S \neq 0 $, and we obtain $ \tilde{B}_{p} \in \mathcal{L}(iH, iV') $ for all $ p $. Since, by assumption, $ \displaystyle \sum\limits_{k \in K} v_{k} B_{k}^{\D} \in \mathcal{L}(iH, iV') $, it follows from the decompositions \eqref{CTLSCC_eq:technoname8} that $ iV \in \mathcal{L}(iH, iV') $ as a linear combination of $ \tilde{B}_{p} $ and $ \displaystyle \sum\limits_{k \in K} v_{k} B_{k}^{\D} $.
    
    This allows us to conclude that
    \begin{equation}
        \mathcal{L}(iH, \, iV') \supset \mathcal{L}(iH, \, iV) \supset \mathfrak{su}(N).
    \end{equation}
    Therefore, the system $ (H, V') $ is also controllable.
\end{proof}

We use this theorem in the form of the following statement.
\begin{corollary*}
    Adding controlled energy-level shifts to the interaction Hamiltonian $ V $ of a $ V $--regular controllable system that initially has no such shifts preserves controllability.
\end{corollary*}

The conditions for three-level systems derived in Section~\ref{CTLSCC_sec:threelevelsystems} illustrate that the addition of controlled energy-level shifts may make a controllable system uncontrollable. Moreover, both the condition that the linear combination of the elements $ B_{k}^{\D}$, describing the original controlled energy-level shifts, belongs to the new dynamical Lie algebra $ \mathcal{L}(iH, \, iV') $, and the $ V$--regularity condition are essential and cannot be omitted. Examples illustrating the necessity of these two conditions are provided in Subsection~\ref{CTLSCC_subsec:highlight}.

\subsection{A basis-element sufficient condition for operator controllability}\label{CTLSCC_subsec:algebrastatement}

To analyze certain classes of three-level systems and determine the structure of the dynamical Lie algebra, we need to identify which elements of the basis $ \mathcal{B} $ belong to this algebra. This raises a question which, although technical, is important: which subset of elements of $ \mathcal{B} $ contained in $ \mathcal{L}(iH, iV) $ is sufficient to conclude that the dynamical Lie algebra contains the entire Lie algebra $\mathfrak{su}(N) $ so that the system is controllable? A possible approach to answering this question is provided by the following theorem \cite{CTLSCC_bib:SchirmerSolomonFu2001}.

\begin{theorem}\label{CTLSCC_th:theoremfromSSF}
    Let $ \mathcal{L} $ be the dynamical Lie algebra of an arbitrary $ N $-level quantum system.  If $ B_{k(k+1)}^{\Re},  B_{k(k+1)}^{\Im} \in \mathcal{L} $ for $ k = 1, \, \ldots, (N-1) $, then $ \mathfrak{su}(N) \subset \mathcal{L} $.  
\end{theorem}
The interpretation of this result is straightforward: if elements $ B_{k(k+1)}^{\Re},  B_{k(k+1)}^{\Im} $ belong to the dynamical Lie algebra $ \mathcal{L}(iH, iV) $, then population transfer can be implemented between each corresponding pair of levels, which ensures controllability of the system. 

Theorem~\ref{CTLSCC_th:theoremfromSSF} is evidently more suitable for the case in which population transfer occurs between adjacent levels; however, this particular choice of elements of the basis $ \mathcal{B} $ is not necessary for establishing controllability. To generalize this result by accommodating alternative choices of these elements, it is convenient to use a graph-based approach.

Let $ \mathcal{G}_{\mathcal{L}} $ be a graph with
\begin{equation}
    V(\mathcal{G}_{\mathcal{L}}) = \left\{1, \, 2, \, \ldots, \, N\right\}, \quad E(\mathcal{G}_{\mathcal{L}}) = \left\{i \leftrightarrow j: \ B_{ij} \in \mathcal{L}\right\}.
\end{equation}
We call an edge $ i \leftrightarrow j \in E(\mathcal{G}_{\mathcal{L}}) $ a \textit{double edge} if both $ B_{ij}^{\Re} \in \mathcal{L} $ and $ B_{ij}^{\Im} \in \mathcal{L} $. 

Consider the following properties of $  \mathcal{G}_{\mathcal{L}} $.
\begin{lemma}\label{CTLSCC_th:lemma1}
    If the graph $\mathcal{G}_{\mathcal{L}} $ is connected, then it is complete.
\end{lemma}
\begin{proof}
    Let $ i $ and $ j $ be distinct vertices of the connected graph $ \mathcal{G}_{\mathcal{L}} $. To prove that $ \mathcal{G}_{\mathcal{L}} $ is complete, we have to show that $ i \leftrightarrow j  \in E\left(\mathcal{G}_{\mathcal{L}} \right)$. By connectivity, either this is trivially true or there exists a path
    \begin{equation}
        i \leftrightarrow k_{1} \leftrightarrow k_{2} \leftrightarrow \ldots \leftrightarrow k_{s} \leftrightarrow j,
    \end{equation}
    where $ k_{p} \in V\left(\mathcal{G}_{\mathcal{L}}\right) $ and $ i \leftrightarrow k_{1}, \, k_{1} \leftrightarrow k_{2}, \, \ldots, \, k_{s} \leftrightarrow j \in E\left(\mathcal{G}_{\mathcal{L}}\right)  $.
    We consider this path consisting of unique vertices $ k_{p} \neq k_{p'} $ for $ p \neq p' $. Otherwise, the cycle between two occurrences of a repeated vertex can be removed to obtain a shorter path. 

    Let $ i \leftrightarrow k_{p} \in E\left(\mathcal{G}_{\mathcal{L}}\right) $. This means that $ B_{ik_{p}} \in \mathcal{L} $ . According to the relations \eqref{CTLSCC_eq:ReRecomm}--\eqref{CTLSCC_eq:ImImcomm},
    \begin{equation}
        [B_{i k_{p}}^{\Im}, \,  B_{k_{p} k_{p+1}}^{\Im}] = -[B_{i k_{p}}^{\Re}, \,  B_{k_{p} k_{p+1}}^{\Re}] = B_{i k_{p+1}}^{\Re} \ \text{ and } \  -[B_{i k_{p}}^{\Re}, \,  B_{k_{p} k_{p+1}}^{\Im}] = -[B_{i k_{p}}^{\Im}, \,  B_{k_{p} k_{p+1}}^{\Re}] = B_{i k_{p+1}}^{\Im}.
    \end{equation}
    Therefore, $ B_{i k_{p+1}} \in \mathcal{L} $ and $ i \leftrightarrow k_{p+1} \in E \left(\mathcal{G}_{\mathcal{L}}\right) $. Applying this argument inductively, with $ j = k_{s+1} $, we obtain $ i \leftrightarrow j \in E(\mathcal{G}_{\mathcal{L}}) $. 
\end{proof}
\begin{lemma}\label{CTLSCC_th:lemma2}
    If the graph  $\mathcal{G}_{\mathcal{L}} $ is complete and has at least one double edge, then all other edges in $ E\left(\mathcal{G}_{\mathcal{L}}\right) $ are double.
\end{lemma}
\begin{proof}
    Let $ i \leftrightarrow j \in E\left(\mathcal{G}_{\mathcal{L}} \right) $ be double. Then
    \begin{equation} \label{CTLSCC_eq:technoname7}
    \frac{1}{2} \left[ B_{ij}^{\Re}, \, B_{ij}^{\Im} \right] = \left(B_{j}^{\D} - B_{i}^{\D} \right) \in \mathcal{L}
    \end{equation}
    By the completeness of the graph, for a vertex $ k \neq i, \, j $, there exists an edge $ j \leftrightarrow k $ and $B_{jk} \in \mathcal{L} $. It is double because
    \begin{equation}
        [B_{jk}^{\Re}, \, (B_{j}^{\D} - B_{i}^{\D})] = B_{jk}^{\Im}, \quad -[ B_{jk}^{\Im}, \, (B_{j}^{\D} - B_{i}^{\D})] = B_{jk}^{\Re}
    \end{equation}
    according to \eqref{CTLSCC_eq:ReDcomm} and \eqref{CTLSCC_eq:ImDcomm}. Making the same steps with $ j $ instead of $ i $, $ k $ instead of $ j $, and a new $ k'$ instead of $ k $, we obtain that an arbitrary edge $ k \leftrightarrow k' $ of $\mathcal{G}_{\mathcal{L}} $ is double.
\end{proof}

If the graph $ \mathcal{G}_{\mathcal{L}} $ is complete and has all edges being double, the algebra $ \mathcal{L} $ includes all elements $ B_{ij}^{\Re} $, $ B_{ij}^{\Im} $, and $ \left(B_j^{\D} - B_i^{\D}\right) $. Thus, combining Lemmas \ref{CTLSCC_th:lemma1} and \ref{CTLSCC_th:lemma2} allows us to conclude that
\begin{theorem}\label{CTLSCC_th:algebrastatement}
    If the graph $ \mathcal{G}_{\mathcal{L}} $ is connected and has at least one double edge, then $ \mathfrak{su}(N) \subset \mathcal{L} $.
\end{theorem}
\begin{corollary*}
    Specifically for $ N = 3 $, this means that
    \begin{equation}
        \exists i, j, k: \  \varepsilon_{ijk} \neq 0 \ \text{ and } \  B_{ij}^{\Re}, B_{ij}^{\Im}, B_{ik} \in \mathcal{L} \ \ \ \Rightarrow \ \ \ \mathfrak{su}(3) \subset \mathcal{L}.
    \end{equation}
\end{corollary*}
Here and below, $ \varepsilon_{ijk} $ denotes the Levi-Civita symbol; thus, $ \varepsilon_{ijk} \neq 0 $ means that the indices $ i $, $ j $, and $ k $ are pairwise distinct.

Theorem \ref{CTLSCC_th:algebrastatement} includes Theorem \ref{CTLSCC_th:theoremfromSSF} as a particular case when the graph $ \mathcal{G}_{\mathcal{L}} $ contains $ (N-1) $ double edges between the $ N $ closest by index neighbor vertices.

\section{Controllability classification of three-level closed quantum systems}\label{CTLSCC_sec:threelevelsystems}

In this section, we obtain a complete operator controllability classification of closed three-level systems driven by a single coherent control. Systems with a free Hamiltonian $ H = c\mathbb{I} $, proportional to the identity matrix, or a disconnected interaction graph are immediately uncontrollable; after separating these cases, we classify the remaining systems according to their energy-level structure, number of allowed transitions, and the presence or absence of controlled energy-level shifts. We then derive necessary and sufficient controllability conditions for systems with two and three allowed transitions, examine the role of the assumptions in Theorem~\ref{CTLSCC_th:shifts} on the addition of controlled energy-level shifts, and summarize the complete classification in Table~\ref{CTLSCC_tab:result}.

\subsection{Structural classification of three-level systems}\label{CTLSCC_subsec:systemtypes}
Before proceeding to the controllability characterization, we first classify three-level systems according to three structural characteristics: their energy-level and level-spacing degeneracies, the number of allowed transitions, and the presence or absence of controlled energy-level shifts. These features are determined by the free Hamiltonian matrix $ H $, the off-diagonal part of the interaction Hamiltonian matrix $ V $, and its diagonal part, respectively.

The general forms of the free Hamiltonian matrix $ H $ and the interaction Hamiltonian matrix $ V $ for such systems are
\begin{equation}\label{CTLSCC_eq:threelevelgeneralmatrices}
    H = \begin{pmatrix}
        E_{1} & 0 & 0 \\
        0 & E_{2} & 0 \\
        0 & 0 & E_{3}
    \end{pmatrix}, \qquad V = \begin{pmatrix}
        v_{1} & v_{12} & v_{13} \\
        v_{12}^{*} & v_{2} & v_{23} \\
        v_{13}^{*} & v_{23}^{*} & v_{3}
    \end{pmatrix},
\end{equation}
where $ E_{k}, v_{k} \in \mathbb{R} $ and $ v_{ij} \in \mathbb{C} $; the free Hamiltonian $ H $ is assumed not to be proportional to the identity matrix, and the nonzero transition matrix elements define a connected structural graph $ \mathcal{G}_{V} $. Table~\ref{CTLSCC_tab:types} describes the key structural differences between three-level systems in terms of their matrices~\eqref{CTLSCC_eq:threelevelgeneralmatrices}.

Based on the energy-level structure determined by the free Hamiltonian $ H $, the systems can be divided into three classes, as illustrated in Figure~\ref{CTLSCC_fig:structures_levels}. The first class consists of systems with neither energy-level nor level-spacing degeneracies; such systems are strongly regular. The second class arises when three distinct energy levels form an equally spaced “ladder”, so that the system is regular but not strongly regular. The last class consists of systems with exactly two degenerate energy levels (the case of three degenerate levels is not considered separately because then the free Hamiltonian $ H $ is proportional to the identity matrix $ \mathbb{I} $).

The form of the off-diagonal part of the interaction Hamiltonian $ V $ determines two classes of systems: those with two and three allowed transitions, respectively. No other numbers of transitions need to be considered: three transitions are the maximum for the three available energy levels, whereas systems with one or no allowed transitions have disconnected structural graphs and for such systems one cannot induce population transfer between all pairs of levels.
\begin{table}[!ht]
    \centering
    \begingroup
    \setlength{\tabcolsep}{4pt}
    \renewcommand{\arraystretch}{1.2}
    \begin{tabular}{|
        >{\centering\arraybackslash}m{0.41\textwidth}|
        >{\centering\arraybackslash}m{0.28\textwidth}|
        >{\centering\arraybackslash}m{0.21\textwidth}|}
    \hline
    \makecell[c]{\raisebox{-0.12cm}{\textbf{By degeneracies}}}
    & \makecell[c]{\textbf{By allowed}\\
        \textbf{transitions}}
    & \makecell[c]{\textbf{By controlled}\\
        \textbf{energy-level shifts}} \\
    \hline

    \makebox[\linewidth][c]{%
        \Gape[0.12cm][0.12cm]{%
            \makecell[c]{Strongly regular [Fig. \ref{CTLSCC_fig:structures_levels}(A)]\\[0.1cm]
                $\exists i, j, k:\ \varepsilon_{ijk} \neq 0,$\\
                $|\Delta E_{ij}| \neq |\Delta E_{jk}| \neq|\Delta E_{ki}| \neq |\Delta E_{ij}|$}}}
    & \multirow[c]{2}{*}[-0.20cm]{%
        \makecell[c]{With two allowed transitions\\[0.1cm]
            $\exists i, j, k:\ \varepsilon_{ijk} \neq 0,$\\
            $v_{ik} = 0,\ v_{ij}, v_{jk} \neq 0$}}
    & \multirow[c]{2}{*}[-0.20cm]{%
        \makecell[c]{Without controlled\\
            energy-level shifts\\[0.1cm]
            $\forall k:\ v_{k} = 0$}} \\
    \cline{1-1}

    \makebox[\linewidth][c]{%
        \Gape[0.12cm][0pt]{%
            \makecell[c]{Regular but not strongly regular
                [Fig. \ref{CTLSCC_fig:structures_levels}(B)]}}}
    &
    & \\
    \cline{2-3}

    \makebox[\linewidth][c]{%
        \Gape[0pt][0.12cm]{%
            \makecell[c]{$\exists i, j, k:\ \varepsilon_{ijk} \neq 0,$\\
                $E_{i} - E_{j} = E_{j} - E_{k}$}}}
    & \multirow[c]{2}{*}[-0.45cm]{%
        \makecell[c]{With three allowed transitions\\[0.1cm]
            $\forall i \neq j:\ v_{ij} \neq 0$}}
    & \multirow[c]{2}{*}[-0.32cm]{%
        \makecell[c]{With controlled\\
            energy-level shifts\\[0.1cm]
            $\exists k:\ v_{k} \neq 0$}} \\
    \cline{1-1}

    \makebox[\linewidth][c]{%
        \Gape[0.12cm][0.12cm]{%
            \makecell[c]{Nonregular [Fig. \ref{CTLSCC_fig:structures_levels}(C)]\\[0.1cm]
                $\exists i \neq j:\ E_{i} = E_{j}$}}}
    &
    & \\
    \hline
    \end{tabular}
    \endgroup
    \caption{Structural classification of three-level systems for the controllability analysis. The systems differ in the degeneracies of their energy levels and transition frequencies determined by the free Hamiltonian $ H $, in the structure of population transfer between levels determined by the off-diagonal part of the interaction Hamiltonian $ V $, and in the controlled energy-level shifts determined by the diagonal part of the interaction Hamiltonian $ V $.}
    \label{CTLSCC_tab:types}
    \vspace{0.5cm}
\end{table}
\begin{figure}[h!]
    \centering
    \includegraphics[width=0.625\linewidth]{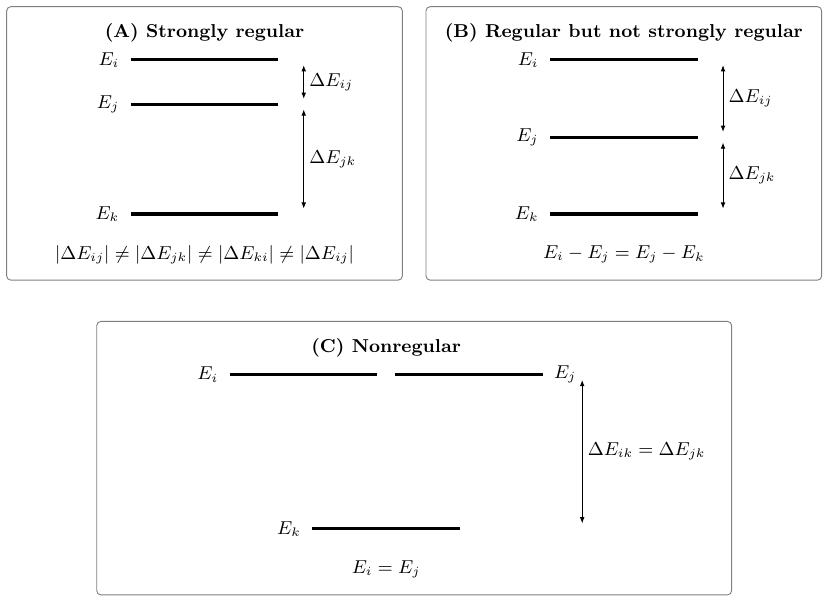}
    \caption{All nontrivial energy-level structures of three-level systems: (A) strongly regular systems, with neither energy-level nor level-spacing degeneracies; (B) regular but not strongly regular systems, with a level-spacing degeneracy but not energy-level degeneracy; (C) nonregular systems, with  two degenerate energy levels. Up to the transformations $H\to aH+c\mathbb I$ with $a\ne0$ and relabeling of the energy levels, these are all possible nontrivial level structures. The diagrams display only the energy-level structure and do not specify allowed transitions or controlled energy-level shifts.}
    \label{CTLSCC_fig:structures_levels}
\end{figure}
\begin{figure}[h!]
    \centering
    \includegraphics[width=0.795\linewidth]{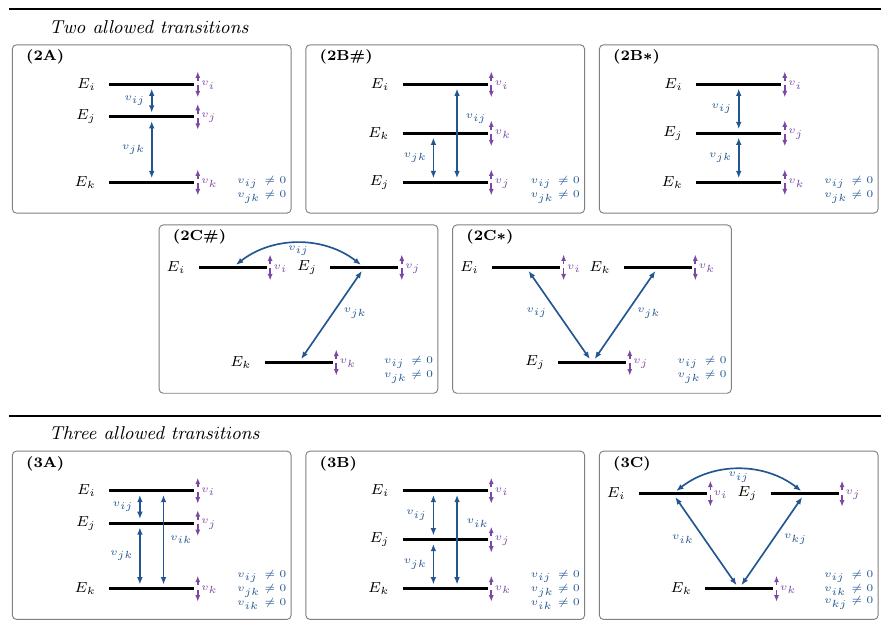}
    \caption{Possible types of three-level systems in terms of their energy-level structure and allowed direct population transfer. The labels identify the systems as follows: the numbers 2 and 3 specify the number of allowed transitions; the symbols A, B, and C specify the energy-level structure; and the symbols $\textasteriskcentered$ and $\#$ denote the presence and absence, respectively, of level-spacing-degenerate transitions.}
    \label{CTLSCC_fig:structures_combinations}
\end{figure}

With respect to the structure of the diagonal part of the interaction Hamiltonian $ V $, three-level systems can be divided into two classes: those with controlled energy-level shifts and those without. Adding controlled energy-level shifts in a number of cases described in Theorem \ref{CTLSCC_th:shifts} preserves controllability, but it may also either render an uncontrollable system controllable or destroy the controllability of an initially controllable system.

For systems with two transitions and a level-spacing degeneracy (and, particularly, an energy-level degeneracy) --- that is, systems that are not strongly regular --- different population transfer schemes are possible: either involving transitions with equal Bohr frequencies or not involving them. These cases should be distinguished since the absence or presence of such transitions is a factor that may affect controllability (e.g., see Theorem \ref{CTLSCC_th:graphtheoremstrengthened}).

Considering jointly the possibilities described above for the energy-level arrangement and the pattern of allowed transitions among them, we distinguish eight distinct cases, each of which may involve either no controlled energy-level shifts or any possible configuration of such shifts. A visual representation of these cases is provided in Figure~\ref{CTLSCC_fig:structures_combinations}. The controllability of each of the systems shown is analyzed as a separate case within the general Theorems~\ref{CTLSCC_th:two_transitions_classification},~\ref{CTLSCC_th:three_transitions_classification}, which are stated and proven below (same as within Theorem~\ref{CTLSCC_th:twonoshifts}, providing already known results).

Some of the systems with two allowed transitions have been explored. The following theorem describes the existing controllability classification for this setup with no controlled energy-level shifts \cite{CTLSCC_bib:KuznetsovPechen2023}. The terminology in \cite{CTLSCC_bib:KuznetsovPechen2023} is slightly different from the terminology of the present work (e.g., a level-spacing degeneracy is defined as level-spacing degenerate transitions in our terms). Therefore, we reformulate the statement in terms of Definitions~\ref{CTLSCC_def:structuralfeatures1},~\ref{CTLSCC_def:structuralfeatures2} and the system classes described in Table \ref{CTLSCC_tab:types}.
\begin{theorem}\label{CTLSCC_th:twonoshifts}
    Let $ (H, \, V) $ be a three-level system with the free Hamiltonian $ H $ and the interaction Hamiltonian $ V $ matrices of the form \eqref{CTLSCC_eq:threelevelgeneralmatrices} with two allowed transitions and no controlled energy-level shifts, i.e.,
    \begin{align}
            &\exists i, j, k: \ \varepsilon_{ijk} \neq 0, \   v_{ik} = 0, \ v_{ij}, v_{jk} \neq 0, \label{CTLSCC_eq:techoldcondition1} \\
            & \forall k': \ v_{k'} = 0 \label{CTLSCC_eq:techoldcondition2} .
    \end{align}
    Then
    \begin{enumerate}
        \item  If the system is strongly regular [Fig.~\ref{CTLSCC_fig:structures_combinations}(2A), no shifts], i.e.,
        \begin{equation}\label{CTLSCC_eq:ttcrossref1}
            |E_{i} - E_{j}| \neq |E_{j} - E_{k}| \neq | E_{i} - E_{k} | \neq |E_{i} - E_{j}|, 
        \end{equation}
        then it is controllable; \label{CTLSCC_enum:techstatement1}
        \item If the system is regular but not strongly regular, and has no level-spacing degenerate transitions [Fig.~\ref{CTLSCC_fig:structures_combinations}(2B\#), no shifts], i.e.,
        \begin{equation}\label{CTLSCC_eq:ttcrossref2}
            E_{i} - E_{k} = E_{k} - E_{j},
        \end{equation}
        then it is controllable; \label{CTLSCC_enum:techstatement2} 
        \item If the system is regular but not strongly regular, and has level-spacing degenerate transitions [Fig.~\ref{CTLSCC_fig:structures_combinations}(2B\textasteriskcentered), no shifts], i.e.,
        \begin{equation}\label{CTLSCC_eq:ttcrossref3}
            E_{i} - E_{j} = E_{j} - E_{k},
        \end{equation}
        then it is controllable iff $ |v_{ij}| \neq |v_{jk}| $; \label{CTLSCC_enum:techstatement3} 
        \item If the system is nonregular and has no level-spacing degenerate transitions [Fig.~\ref{CTLSCC_fig:structures_combinations}(2C\#), no shifts], i.e.,
        \begin{equation}\label{CTLSCC_eq:ttcrossref4}
            E_{i} = E_{j} \quad (\text{or } \, E_{j} = E_{k}),
        \end{equation}
        then it is controllable;
        \item If the system is nonregular and has level-spacing degenerate transitions [Fig.~\ref{CTLSCC_fig:structures_combinations}(2C\textasteriskcentered), no shifts], i.e.,
        \begin{equation}\label{CTLSCC_eq:ttcrossref5}
            E_{i} = E_{k},
        \end{equation}
        then it is \underline{not} controllable.
    \end{enumerate}
\end{theorem}

In the following Subsections~\ref{CTLSCC_subsec:twotransitions},~\ref{CTLSCC_subsec:threetransitions}, we proceed from these particular conditions to a complete characterization of controllability of all three-level systems. First, we extend the statement of Theorem~\ref{CTLSCC_th:twonoshifts} to include controlled energy-level shifts. Then, we explore the systems with three allowed transitions.

\subsection{Systems with two allowed transitions}\label{CTLSCC_subsec:twotransitions}

In this subsection, we investigate controllability of the systems with two allowed transitions. We derive controllability conditions for the five cases shown in Figure~\ref{CTLSCC_fig:structures_combinations} and in Theorem~\ref{CTLSCC_th:twonoshifts}, that allow the inclusion of controlled energy-level shifts $ v_{k} $, thereby omitting the assumption \eqref{CTLSCC_eq:techoldcondition2}. Theorem~\ref{CTLSCC_th:two_transitions_classification} extends Theorem~\ref{CTLSCC_th:twonoshifts} to arbitrary controlled energy-level shifts; when $v_k=0$ for all $k$, it reduces to the known classification stated in Theorem~\ref{CTLSCC_th:twonoshifts}.

\begin{theorem}[Complete operator controllability classification for systems with two allowed transitions]\label{CTLSCC_th:two_transitions_classification}
\hspace{0.1cm} \\[0.05cm]
\indent Let $ (H, V) $ be a three-level system with the free Hamiltonian $ H $ and the interaction Hamiltonian $ V $ matrices of the form \eqref{CTLSCC_eq:threelevelgeneralmatrices} with two allowed transitions \eqref{CTLSCC_eq:techoldcondition1}. Then the following statements hold.
\begin{itemize}
    \item \textbf{Controllability Condition 1}: 
    
    If the system is strongly regular [Fig.~\ref{CTLSCC_fig:structures_combinations}(2A)], i.e., if the inequalities \eqref{CTLSCC_eq:ttcrossref1} hold, then it is controllable.

    \item \textbf{Controllability Condition 2}: 
    
    If the system is regular but not strongly regular and has no level-spacing degenerate transitions [Fig.~\ref{CTLSCC_fig:structures_combinations}(2B\#)], i.e., if the equality \eqref{CTLSCC_eq:ttcrossref2} holds, then it is controllable.

    \item \textbf{Controllability Condition 3}: 
    
    If the system is regular but not strongly regular and has level-spacing degenerate transitions [Fig.~\ref{CTLSCC_fig:structures_combinations}(2B\textasteriskcentered)],\\ i.e., if the equality \eqref{CTLSCC_eq:ttcrossref3} holds, then it is controllable iff 
    \begin{equation}
        |v_{ij}|\neq|v_{jk}|
        \quad\text{or}\quad
        v_j-v_i\neq v_k-v_j.
    \end{equation}
     \item \textbf{Controllability Condition 4}: 
     
     If the system is nonregular and has no level-spacing degenerate transitions [Fig.~\ref{CTLSCC_fig:structures_combinations}(2C\#)], i.e., if the equality \eqref{CTLSCC_eq:ttcrossref4} holds, then it is controllable.

    \item \textbf{Controllability Condition 5}: 
    
    If the system is nonregular and has level-spacing degenerate transitions [Fig.~\ref{CTLSCC_fig:structures_combinations}(2C\textasteriskcentered)], i.e., if the equality \eqref{CTLSCC_eq:ttcrossref5} holds, then it is controllable iff
    \begin{equation}
        v_i\neq v_k.
    \end{equation}
\end{itemize}
\end{theorem}

\begin{clproof} \label{CTLSCC_th:twostrong}
    This statement is a direct corollary of the sufficient controllability condition in Theorem \ref{CTLSCC_th:graphtheorem} or a combination of Theorems \ref{CTLSCC_th:shifts} and \ref{CTLSCC_th:twonoshifts}.
\end{clproof}

\begin{clproof}\label{CTLSCC_th:twonotstrong1}
    This statement is a direct corollary of Theorem \ref{CTLSCC_th:graphtheoremstrengthened} or a combination of Theorems \ref{CTLSCC_th:shifts} and \ref{CTLSCC_th:twonoshifts}.
\end{clproof}

\begin{clproof}\label{CTLSCC_th:twonotstrong2}
    Without loss of generality, we consider energy levels in the order
    \begin{equation}\label{CTLSCC_eq:reindexation}
        i = 1, \, j = 2, \, k = 3,
    \end{equation}
     so that
    \begin{align}
            &v_{12}, \, v_{23} \neq 0, \quad v_{13} = 0, \label{CTLSCC_eq:twotransitionsindexation} \\[0.2cm]
            &E_{1} = E, \quad  E_{2} = E + \Delta E, \, E_{3} = E + 2 \Delta E. %\label{STLSCC_eq:regnotstronglylevels}
    \end{align}
    Thus, we can modify the free Hamiltonian by a constant shift and multiplication as  $ (\Delta E)^{-1} (iH - (E + \Delta E) \, i\mathbb{I})$ to produce the modified free Hamiltonian which generates the same dynamical Lie algebra
\begin{equation}\label{CTLSCC_eq:regnotstronglyreparametrization}
        iH = B_{3}^{\D} - B_{1}^{\D} = i\begin{pmatrix}
            -1 & 0 & 0 \\
            0 & 0 & 0 \\
            0 & 0 & 1 
        \end{pmatrix}.
    \end{equation}
    
    As in Controllability Condition~\ref{CTLSCC_th:twonotstrong1}, the criterion presented in Theorem~\ref{CTLSCC_th:twonoshifts} and Corollary of Theorem~\ref{CTLSCC_th:shifts} together show that the system is controllable when $ |v_{12}| \neq |v_{23} | $, independently of the values of $ v_{k} $. Therefore, it remains to perform the analysis for the case $ |v_{12}| = |v_{23}| =:  r $, when the interaction Hamiltonian has the form    \begin{equation}\label{CTLSCC_eq:techproof10formula1}
        iV = \begin{pmatrix}
            iv_{1} & r e^{i (\varphi_{12} + \frac{\pi}{2})} & 0 \\
            r e^{-i \left(\varphi_{12} - \frac{\pi}{2} \right)} & iv_{2} & r e^{i (\varphi_{23} + \frac{\pi}{2})} \\
            0 & r e^{-i (\varphi_{23} - \frac{\pi}{2})}  & iv_{3}
        \end{pmatrix}.
    \end{equation}
    For matrices \eqref{CTLSCC_eq:regnotstronglyreparametrization} and \eqref{CTLSCC_eq:techproof10formula1} one has
    \begin{equation}\label{CTLSCC_eq:diagonalfinding}
        L := i V + \ad_{i H}^{2}(iV) = \sum\limits_{k = 1}^{3} v_{k} B_{k}^{\D} = i \begin{pmatrix}
            v_{1} & 0 & 0 \\
            0 & v_{2} & 0 \\
            0 & 0 & v_{3}
        \end{pmatrix} \in \mathcal{L}(iH, iV).
    \end{equation}
    If $ (v_{2}-v_{1}) \neq (v_{3} - v_{2}) $, then the set $ \left\{iH, \, L, \, i\mathbb{I} \right\} \subset \mathcal{L}(iH, \, iV) $ is linearly independent. Hence, $ \exists  c_{1}, \, c_{2}, \, c_{3} \in \mathbb{R}:$  
    \begin{equation}\label{CTLSCC_eq:efficientstructure}
        \begin{split}
            &i H' = c_{1} i H + c_{2} L + c_{3} i \mathbb{I}, \\[0.2cm]
            &|E'_{2} - E'_{1}| \neq |E'_{3} - E'_{2}| \neq |E'_{3} - E'_{1}| \neq |E'_{2} - E'_{1}|. 
        \end{split}
    \end{equation}
    Therefore, the dynamical Lie algebra $ \mathcal{L}(iH', \, iV) $ corresponds to a system that is controllable by Controllability Condition \ref{CTLSCC_th:twostrong}, and moreover $ \mathcal{L}(iH, iV) \supset \mathcal{L}(iH', iV) $. This proves the sufficiency of the condition.
    
    Necessity follows from the fact that, when $ v_{2} - v_{1} = v_{3} - v_{2}$, one can consider an equivalent interaction Hamiltonian $ V' $ with a symmetric shift of all levels by $ v_{2} $ and with the same transition structure, 
    \begin{equation}\label{CTLSCC_eq:newvstrih}
        iV' := iV + c iH = \begin{pmatrix}
            iv_{2} & r e^{i (\varphi_{12} + \frac{\pi}{2})} & 0 \\
            r e^{-i \left(\varphi_{12} - \frac{\pi}{2} \right)} & iv_{2} & r e^{i (\varphi_{23} + \frac{\pi}{2})} \\
            0 & r e^{-i (\varphi_{23} - \frac{\pi}{2})}  & iv_{2}
        \end{pmatrix},
    \end{equation}
    where
    \begin{equation}
        c = v_{1} - v_{2} = v_{2} - v_{3} = \frac{v_{1} - v_{3}}{2}.
    \end{equation}
    Indeed, for an arbitrary element $ A $ of the dynamical Lie algebra $ \mathcal{L}(iH, iV') $, one has
    \begin{equation}
        [A, \, iV'] = c[A, \, iH] + [A, \, iV],
    \end{equation}
    and hence it follows that $ \mathcal{L}(iH, iV) \supset \mathcal{L}(iH, iV') $. Rewriting the same equality for $ A \in \mathcal{L}(iH, iV) $ shows the reverse inclusion $ \mathcal{L}(iH, iV') \supset \mathcal{L}(iH, iV) $. Consequently, $ \mathcal{L}(iH, iV) = \mathcal{L}(iH, iV')$.
    
    However, a system with symmetric controlled energy-level shifts is equivalent, in its controllability, to a system without such shifts. That is, the system with interaction Hamiltonian $ V' $ is equivalent to the system with the interaction Hamiltonian $  V'' := V' - v_{2} \mathbb{I} $. In this case, we have 
    \begin{equation}
        iV'' = iV' - v_{2} i\mathbb{I} = \begin{pmatrix}
            0 & r e^{i (\varphi_{12} + \frac{\pi}{2})} & 0 \\
             r e^{-i (\varphi_{12} - \frac{\pi}{2})} & 0 & r e^{i (\varphi_{23} + \frac{\pi}{2})} \\
             0 & r e^{-i (\varphi_{23} - \frac{\pi}{2})}  & 0
        \end{pmatrix}. 
    \end{equation}
    But such a system $ (H, V'') $ is not controllable by the case~\ref{CTLSCC_enum:techstatement3} in Theorem~\ref{CTLSCC_th:twonoshifts}.
\end{clproof}

\begin{clproof}\label{CTLSCC_th:twonotregular1}
    We again use the indexing \eqref{CTLSCC_eq:reindexation}, \eqref{CTLSCC_eq:twotransitionsindexation} of levels and corresponding transitions. Since
    \begin{equation}
        E_{1} = E_{2} = E, \quad E_{3} = E', \quad E \neq E',
    \end{equation}
    without loss of generality, we can shift and rescale $ i H $ to $ (E'- E)^{-1} (iH - E \, i \mathbb{I}) $, so that it takes the form
    \begin{equation*}
        iH = B_{3}^{\D} = i \begin{pmatrix}
            0 & 0 & 0 \\
            0 & 0 & 0 \\
            0 & 0 & 1
        \end{pmatrix}.
    \end{equation*}
    
    Then we have
    \begin{align}
        & L_{1} := - \ad_{iH}^{2}(iV) = v_{23}^{\Im} B_{23}^{\Re} + v_{23}^{\Re} B_{23}^{\Im} = 
        i \begin{pmatrix}
            0 & 0 & 0 \\
            0 & 0 & v_{23} \\
            0 & v_{23}^{*}  & 0 
        \end{pmatrix}, \\[0.2cm]
        & L_{2} := [i H, \, L_{1}] = - v_{23}^{\Re} B_{23}^{\Re} + v_{23}^{\Im} B_{23}^{\Im} = \begin{pmatrix}
            0 & 0 & 0 \\
            0 & 0 & v_{23} \\
            0 & -v_{23}^{*}  & 0 
        \end{pmatrix}.
    \end{align}
    Thus, for $ L = \begin{pmatrix}
        L_{1} & L_{2}
    \end{pmatrix}^{T} $ and $ B = \begin{pmatrix}
        B_{23}^{\Re} & B_{23}^{\Im}
    \end{pmatrix}^{T} $, in Lemma~\ref{CTLSCC_th:lemmainvertiblebasis}, we have
    \begin{equation}\label{CTLSCC_eq:technoname10}
        S = \begin{pmatrix}
            v_{23}^{\Im} & v_{23}^{\Re}  \\
            -v_{23}^{\Re} & v_{23}^{\Im}  
        \end{pmatrix}, \quad \det S = (v_{23}^{\Re})^{2} + (v_{23}^{\Im})^{2} = |v_{23}|^{2} \neq 0. 
    \end{equation}
    It shows that $ B_{23}^{\Re}, B_{23}^{\Im} \in \mathcal{L}(iH, \, iV) $.

    Next, the dynamical Lie algebra also contains elements
    \begin{equation}
    \begin{split}
        & L_{3} := [ B_{23}^{\Im}, \, [B_{23}^{\Re}, \, iV]] - 4 v_{23}^{\Re} B_{23}^{\Re} = - v_{12}^{\Re} B_{12}^{\Re} + v_{12}^{\Im} B_{12}^{\Im} = \begin{pmatrix}
            0 & v_{12} & 0 \\
            -v^{*}_{12} & 0 & 0 \\
            0 & 0 & 0
        \end{pmatrix}, \\[0.2cm]
        & L_{4} := -[B_{23}^{\Im}, \, [B_{23}^{\Re}, \, L_{3}]] = v_{12}^{\Im} B_{12}^{\Re} + v_{12}^{\Re} B_{12}^{\Im} = 
            \begin{pmatrix}
                0 & i v_{12} & 0 \\
                i v_{12}^{*} & 0 & 0  \\
                0 & 0 & 0 
            \end{pmatrix}.
    \end{split}
    \end{equation}
    By constructing the matrix $ S $ in the same manner as in \eqref{CTLSCC_eq:technoname10}, we find that $ B_{12}^{\Re}, B_{12}^{\Im} \in \mathcal{L}(iH, iV) $. Therefore, the system is controllable by Theorem \ref{CTLSCC_th:algebrastatement}.
\end{clproof}

\begin{clproof}\label{CTLSCC_th:twonotregular2}
    Again, we use the indexing of the levels and transitions given in \eqref{CTLSCC_eq:reindexation} and \eqref{CTLSCC_eq:twotransitionsindexation}. It gives
    \begin{equation}
        E_{1} = E_{3} = E, \ \ E_{2} = E' \ \ \ \ (E \neq E').
    \end{equation}
    We can use the modified form $ (E' - E)^{-1} (iH - E \, i \mathbb{I}) $ of $ i H $ without affecting the controllability of the system. Therefore, to prove the theorem, we can take 
    \begin{equation}\label{CTLSCC_eq:notregreparametrization}
        iH = B_{2}^{\D} = i \begin{pmatrix}
            0 & 0 & 0 \\
            0 & 1 & 0 \\
            0 & 0 & 0
        \end{pmatrix}.
    \end{equation}

Consider $v_1\ne v_3$. The diagonal element $L$ defined in \eqref{CTLSCC_eq:diagonalfinding}  belongs to $\mathcal L(iH,iV)$. Choose $a\in\mathbb R$ such that $v_1$, $v_2+a$, and $v_3$ are pairwise distinct and have pairwise distinct absolute differences. Such a choice exists because only finitely many values of $a$ are excluded. Then $iH'=L+aiH$ belongs to $\mathcal L(iH,iV)$, while $H'$ is strongly regular. Since the interaction graph is connected, the system $(H',V)$ is controllable by Theorem~\ref{CTLSCC_th:graphtheorem}. Hence the original system $(H,V)$ is controllable, and the sufficiency is established.

Now consider $v_1=v_3$. As in \eqref{CTLSCC_eq:newvstrih}, set $ iV':=iV+c\,iH$ with
    \begin{equation}
        c = v_{3} - v_{2} = v_{1} - v_{2}.
    \end{equation}
Then $iV'$ has equal diagonal elements, while $\mathcal L(iH,iV')=\mathcal L(iH,iV)$. After subtracting this common diagonal component, the system has no controlled energy-level shifts and is noncontrollable by item~5 of Theorem~\ref{CTLSCC_th:twonoshifts}. Hence the system $(H,V)$ is not controllable, and the necessity is established. 
\end{clproof}

\subsection{Systems with three allowed transitions}\label{CTLSCC_subsec:threetransitions}
In this subsection, we perform the controllability analysis of systems with three allowed transitions. To the best of our knowledge, no complete classification of such systems without controlled energy-level shifts has been reported in the literature; therefore, we perform the controllability analysis for the systems both without and with controlled energy-level shifts. Unlike systems with two transitions, systems with three allowed transitions with degeneracies are not divided into two cases depending on the presence or absence of level-spacing degenerate transitions. Therefore, we have three statements instead of five in the previous subsection --- one for each energy-level structure, shown in Figure~\ref{CTLSCC_fig:structures_combinations}.

\begin{theorem}[Complete operator controllability classification for systems with three allowed transitions]\label{CTLSCC_th:three_transitions_classification}
    \hspace{0.1cm} \\[0.05cm]
    \indent Let $ (H, V) $ be a three-level system with the free Hamiltonian $ H $ and the interaction Hamiltonian $ V $ matrices of the form \eqref{CTLSCC_eq:threelevelgeneralmatrices} with three allowed transitions, i.e., 
    \begin{equation}
        \forall i' \neq j': \ v_{i'j'} \neq 0.
    \end{equation}
     Then the following statements hold.
\begin{itemize}
    \item \textbf{Controllability Condition 6}: 

    If the system is strongly regular  [Fig.~\ref{CTLSCC_fig:structures_combinations}(3A)], i.e., if \, $ \exists i, j, k: \, \varepsilon_{ijk} \neq 0 $ for which the inequalities \eqref{CTLSCC_eq:ttcrossref1} hold, then it is controllable.

    \item \textbf{Controllability Condition 7}: 

    If the system is regular but not strongly regular [Fig.~\ref{CTLSCC_fig:structures_combinations}(3B)], i.e., if \, $ \exists i, j, k: \, \varepsilon_{ijk} \neq 0 $ for which the equality~\eqref{CTLSCC_eq:ttcrossref3} holds, then it is controllable.

    \item \textbf{Controllability Condition 8}: 
    
    If the system is nonregular [Fig.~\ref{CTLSCC_fig:structures_combinations}(3C)], i.e., if $ \exists i, j, k: \  \varepsilon_{ijk} \neq 0 $ for which 
    \begin{equation}
          E_{i} = E_{j},
    \end{equation}
    then it is controllable iff one of the following conditions holds:
     \begin{enumerate}
         \item \begin{equation}
             \frac{v_{ij}}{v_{ik} v_{kj}} \notin \mathbb{R}, \ \text{i.e.} \ \arg(v_{ij}) \neq \arg(v_{ik}) + \arg(v_{kj}) + \pi m,\quad m\in\mathbb Z;
         \end{equation}
         \item
         \begin{equation}
             \frac{v_{ij}}{v_{ik} v_{kj}}  \left(|v_{kj}|^{2} - |v_{ik}|^{2} \right) \neq v_{j} - v_{i}.
         \end{equation}
     \end{enumerate}
\end{itemize}
\end{theorem}

\begin{clproof}\label{CTLSCC_th:threestrong}
    This is a direct corollary of the sufficient controllability condition in Theorem \ref{CTLSCC_th:graphtheorem}.
\end{clproof}

\begin{clproof}\label{CTLSCC_th:threenotstrong}
    First, if the statement of the theorem is true for systems without controlled energy-level shifts, then it is true for any configuration of the additional controlled energy-level shifts due to Theorem \ref{CTLSCC_th:shifts}, since the considered system is regular and hence $V$-regular. For this reason, it is enough to consider the case $ \forall k: \ v_k = 0 $.

    As in Controllability Condition~\ref{CTLSCC_th:twonotstrong2}, we can, without loss of generality, take the free Hamiltonian $ H $ to have the form given in the equation~\eqref{CTLSCC_eq:regnotstronglyreparametrization}. Since $ v_{ij} \neq 0 $, the structure of the interaction Hamiltonian $ V $ allows us to perform the same reparameterization as before. Since we have excluded controlled energy-level shifts,
    \begin{equation}\label{CTLSCC_eq:technoname2}
        iV = \sum\limits_{i < j} \left(v_{ij}^{\Im} B_{ij}^{\Re} + v_{ij}^{\Re} B_{ij}^{\Im}\right) = i \begin{pmatrix}
            0 & v_{12} & v_{13} \\
            v_{12}^{*} & 0 & v_{23} \\
            v_{13}^{*} & v_{23}^{*} & 0
        \end{pmatrix},
    \end{equation}
    where $ v_{12}, v_{13}, v_{23} \neq 0 $.

    First, the dynamical Lie algebra contains 
    \begin{equation}\label{CTLSCC_eq:technoname1.1}
        L_{1} := \frac{1}{3}\left(\ad_{iH}^{2}(i V) + 4 i V\right) = v_{12}^{\Im} B_{12}^{\Re} + v_{12}^{\Re} B_{12}^{\Im} + v_{23}^{\Im} B_{23}^{\Re} + v_{23}^{\Re} B_{23}^{\Im} = i \begin{pmatrix}
           0 & v_{12} & 0 \\
           v_{12}^{*} & 0 & v_{23} \\
           0 & v_{23}^{*} & 0 \\
        \end{pmatrix}.
    \end{equation}
    Therefore, by the regular but not strongly regular condition in Theorem \ref{CTLSCC_th:twonoshifts}, the system is controllable at least whenever $ |v_{12}| \neq |v_{23}| $ as
    \begin{equation}\label{CTLSCC_eq:technoname1.2}
        \mathcal{L}(iH, iV) \supset \mathcal{L}(iH, L_{1}). 
    \end{equation}
    However, the following shows that the third allowed transition $ v_{13} \neq 0 $ removes this restriction.

    Indeed, since $ L_{1} \in \mathcal{L}(iH, iV) $, we can also obtain 
    \begin{align}
        L_{2} &:= i V - L_{1}
        = v_{13}^{\Im} B_{13}^{\Re} + v_{13}^{\Re} B_{13}^{\Im}
        = i \begin{pmatrix}
            0 & 0 & v_{13} \\
            0 & 0 & 0 \\
            v_{13}^{*} & 0 & 0
        \end{pmatrix}, \\[1em]
        L_{3} &:= \frac{1}{2}[i H, \, L_{2}]
            = -v_{13}^{\Re} B_{13}^{\Re} + v_{13}^{\Im} B_{13}^{\Im}
            = \begin{pmatrix}
            0 & 0 & v_{13} \\
            0 & 0 & 0 \\
            - v_{13}^{*} & 0 & 0
        \end{pmatrix}.
    \end{align}
    Constructing $ S $ for $ L = \begin{pmatrix}
        L_{2} & L_{3}
    \end{pmatrix}^{T} $ and $ B = \begin{pmatrix}
        B_{13}^{\Re} & B_{13}^{\Im}
    \end{pmatrix}^{T}$, we have
    \begin{equation}
        S = \begin{pmatrix}
            v_{13}^{\Im} & v_{13}^{\Re} \\
            -v_{13}^{\Re} & v_{13}^{\Im}
        \end{pmatrix}, \quad \det S = (v_{13}^{\Re})^{2}+ (v_{13}^{\Im})^{2} = |v_{13}|^{2} \neq 0.
    \end{equation}
    This means that $ B_{13}^{\Re}, B_{13}^{\Im} \in \mathcal{L}(iH, iV) $.

    Next, the commutator
    \begin{equation}
        L_{4} := [i H, \, L_{1}]= -v_{12}^{\Re} B_{12}^{\Re} + v_{12}^{\Im} B_{12}^{\Im} - v_{23}^{\Re} B_{23}^{\Re} + v_{23}^{\Im} B_{23}^{\Im} = \begin{pmatrix}
            0 & v_{12} & 0 \\
            -v_{12}^{*}& 0 & v_{23} \\
            0 & -v_{23}^{*} & 0\\
        \end{pmatrix}
    \end{equation}
    belongs to the dynamical Lie algebra. But since it also contains the basis element $ B_{13}^{\Im} $ (the basis element $ B_{13}^{\Re} $ could also be used in an analogous calculation, therefore it is crucial that $ B_{13} \in \mathcal{L}(iH, \,iV)$), it contains the elements
    \begin{align}
        L_{5} &:= -[L_{1}, \, B_{13}^{\Im}]
        = v_{23}^{\Re} B_{12}^{\Re} + v_{23}^{\Im} B_{12}^{\Im}
         - v_{12}^{\Re} B_{23}^{\Re} - v_{12}^{\Im} B_{23}^{\Im}
        = \begin{pmatrix}
        0 & -v_{23}^{*} & 0 \\
            v_{23} & 0 & v_{12}^{*} \\
            0 & -v_{12} & 0 
        \end{pmatrix}, \\[1em]
    L_{6} &:= [i H, \, L_{5}]
        = -v_{23}^{\Im} B_{12}^{\Re} + v_{23}^{\Re} B_{12}^{\Im}
        + v_{12}^{\Im} B_{23}^{\Re} - v_{12}^{\Re} B_{23}^{\Im}
        = i \begin{pmatrix}
            0 & v_{23}^{*} & 0 \\
            v_{23} & 0 & -v_{12}^{*} \\
            0 & -v_{12} & 0 
    \end{pmatrix}.
\end{align}
    Repeating the construction of the matrix $ S $ for $ L = \begin{pmatrix}
        L_{1} & L_{4} & L_{5} & L_{6} 
    \end{pmatrix}^{T} $ and $ B = \begin{pmatrix}
        B_{12}^{\Re} & B_{12}^{\Im} & B_{23}^{\Re} & B_{23}^{\Im}
    \end{pmatrix}^{T} $, we obtain
    \begin{equation}
        S = \begin{pmatrix}
           v_{12}^{\Im} & v_{12}^{\Re} & v_{23}^{\Im} & v_{23}^{\Re} \\
           -v_{12}^{\Re} & v_{12}^{\Im} & -v_{23}^{\Re} & v_{23}^{\Im} \\
           v_{23}^{\Re} & v_{23}^{\Im} & -v_{12}^{\Re} & -v_{12}^{\Im} \\
            -v_{23}^{\Im} & v_{23}^{\Re} & v_{12}^{\Im} & -v_{12}^{\Re} \\
        \end{pmatrix},
    \end{equation}
    \begin{equation}
        \det S = \left((v_{12}^{\Re})^{2} + (v_{12}^{\Im})^{2} + (v_{23}^{\Re})^{2} + (v_{23}^{\Im})^{2}\right)^{2} = \left(|v_{12}|^{2} + |v_{23}|^{2}\right)^{2} \neq 0.
    \end{equation}
    Therefore, $ B_{12}^{\Re}, B_{12}^{\Im}, \,B_{23}^{\Re}, \, B_{23}^{\Im} \in \mathcal{L}(iH, \, iV) $.

    We have shown that $ \forall i < j: \ B_{ij}^{\Re}, B_{ij}^{\Im} \in \mathcal{L}(iH, \, iV) $. Thus, the system is controllable by Theorem~\ref{CTLSCC_th:algebrastatement}.
\end{clproof}

\begin{clproof}\label{CTLSCC_th:threenotregular}
    Here, we use the indexing 
    \begin{equation}
        i = 1, \ j = 3, \ k = 2
    \end{equation}
    and, consequently, the free Hamiltonian $ H $ of the form \eqref{CTLSCC_eq:notregreparametrization} as in Controllability Condition \ref{CTLSCC_th:twonotregular2}. Then we have 
    \begin{align}
        L_{1} &:= [i H, \, i V]
        = -v_{12}^{\Re} B_{12}^{\Re} + v_{12}^{\Im} B_{12}^{\Im}
        + v_{23}^{\Re} B_{23}^{\Re} - v_{23}^{\Im} B_{23}^{\Im}
        = \begin{pmatrix}
            0 & v_{12} & 0 \\
            -v_{12}^{*}& 0 & -v_{23} \\
            0 & v_{23}^{*}& 0
        \end{pmatrix}, \label{CTLSCC_eq:L1} \\[1em]
        L_{2} &:= -\ad_{i H}^{2}(iV)
        = v_{12}^{\Im} B_{12}^{\Re} + v_{12}^{\Re} B_{12}^{\Im}
        + v_{23}^{\Im} B_{23}^{\Re} + v_{23}^{\Re} B_{23}^{\Im}
        = i \begin{pmatrix}
        0 & v_{12} & 0 \\
        v_{12}^{*} & 0 & v_{23} \\
        0 & v_{23}^{*} & 0
        \end{pmatrix},  \label{CTLSCC_eq:L2}\\[1em]
        L_{3} &:= iV - L_{2} - v_{2}\, i H
        = v_{13}^{\Im} B_{13}^{\Re} + v_{13}^{\Re} B_{13}^{\Im}
        + v_{1} B_{1}^{\D} + v_{3} B_{3}^{\D}
        = i \begin{pmatrix}
            v_{1} & 0 & v_{13} \\
            0 & 0 & 0 \\
            v_{13}^{*} & 0 & v_{3}
        \end{pmatrix}, \label{CTLSCC_eq:L3}
    \end{align}
    and 
    \begin{equation}
        \begin{split}
            L_{4} &:= \frac{1}{2}[L_{1}, L_{2}] + \left(|v_{12}|^{2} + |v_{23}|^{2}\right) \, i H = \Im(v_{12} v_{23}) B_{13}^{\Re} + \Re(v_{12} v_{23}) B_{13}^{\Im} + |v_{12}|^{2} B_{1}^{\D} + |v_{23}|^{2} B_{3}^{\D} = \\
            & = i \begin{pmatrix} |v_{12}|^{2} & 0 & v_{12} v_{23} \\
            0 & 0 & 0 \\
            v_{12}^{*} v_{23}^{*} & 0 & |v_{23}|^{2}
            \end{pmatrix}.
        \end{split}
    \end{equation}
    These are the key elements of the dynamical Lie algebra we are going to analyze. Their commutators with $ iH $ produce no new elements: $ L_{1} $ and $ L_{2} $ are mapped to each other up to sign, whereas $ L_{3} $ and $ L_{4} $ commute with $ iH $.
    
    We consider several possible cases separately. The following cases \ref{CTLSCC_enum:1} and \ref{CTLSCC_enum:2} are related to sufficiency of the criterion (i.e., proving the system is controllable when the condition is satisfied), while the case \ref{CTLSCC_enum:3} describes necessity (i.e., proving the system is not controllable when the condition is not satisfied).
    \begin{enumerate}
        \item $ \displaystyle  \frac{v_{13}}{v_{12} v_{23}} \notin \mathbb{R} \ \ (\arg(v_{13}) \neq \arg(v_{12}) + \arg(v_{23}) + \pi m) $:\label{CTLSCC_enum:1} 
        \vspace{0.1cm}
        
        We start with the case of inconsistent phases of the transitions that break the symmetry of the structure of the interaction Hamiltonian $ V $. It includes two different situations.
        \begin{enumerate}[label=\arabic{enumi}.\arabic*]
            \item $ |v_{12}| \neq |v_{23}| $ or $ v_{1} \neq v_{3} $:\label{CTLSCC_enum:1.1}
                \vspace{0.1cm}
                
            Consider also the following two elements of the dynamical Lie algebra $ \mathcal{L}(iH, \, iV) $
            \begin{align}
                L_{5} &:= (v_{1} |v_{23}|^{2} - v_{3} |v_{12}|^{2}) (iH - i \mathbb{I}) - (v_{3} - v_{1}) L_{4} + \left(|v_{23}|^{2} - |v_{12}|^{2}\right) L_{3} = \\[0.2cm]
                &= w_{13}^{\Im} B_{13}^{\Re} + w_{13}^{\Re} B_{13}^{\Im} = i \begin{pmatrix}
                    0 & 0 & w_{13} \\
                    0 & 0 & 0 \\
                    w_{13}^{*} & 0 & 0
                    \end{pmatrix}, \nonumber \\[0.2cm]
                L_{6} &:= -[L_{3}, \, L_{4}] + w \, (i H -  i \mathbb{I}) = \\
                &= -w_{13}^{\Re} B_{13}^{\Re} + w_{13}^{\Im} B_{13}^{\Im} - 2 w B_{3}^{\D} = \begin{pmatrix} 
                    0 & 0 & w_{13} \\
                    0 & 0 & 0 \\
                    -w_{13}^{*} & 0 & -2 i w
                \end{pmatrix}, \nonumber
            \end{align}
            where 
            \begin{equation}
                \begin{split}
                    &w_{13} = \left(|v_{23}|^{2} - |v_{12}|^{2}\right) v_{13} -(v_{3} - v_{1}) v_{12} v_{23}, \quad w_{13}^{\Re} = \Re(w_{13}), \, w_{13}^{\Im} = \Im(w_{13}), \\[0.2cm]
                    &w = -2 \, \Im(v_{12} v_{23} v_{13}^{*}). 
                \end{split}
            \end{equation}
            
            Conditions in the case \ref{CTLSCC_enum:1} and in the subcase \ref{CTLSCC_enum:1.1} together guarantee that $ w_{13} \neq 0 $. The second condition implies that at least one of the terms is nonzero, whereas the first ensures that they do not cancel each other out in the case of both being nonzero, since they possess distinct phases. Adding a multiple of $ i \mathbb{I} $, as has been mentioned before, does not affect the dynamical Lie algebra structure.
            
            Then,
            \begin{equation}
                \begin{split}
                    L_{7} &:= \frac{1}{2} [L_{5}, L_{6}] - |w_{13}|^{2}(i H - i \mathbb{I}) = - w \, w_{13}^{\Re} B_{13}^{\Re} + w \, w_{13}^{\Im} B_{13}^{\Im} + 2 |w_{13}|^{2} B_{3}^{\D} = \\[0.2cm]  
                    &= \begin{pmatrix}
                        0 & 0 & w \, w_{13} \\
                        0 & 0 & 0 \\
                        - w \, w_{13}^{*} & 0 & 2 i |w_{13}|^{2}
                    \end{pmatrix}.
                \end{split}
            \end{equation}
            Taking $ S $ for $ L = \begin{pmatrix}
                 L_{5} & L_{6} & L_{7}
             \end{pmatrix}^{T} $ and $ B = \begin{pmatrix}
                 B_{13}^{\Re} & B_{13}^{\Im} & B_{3}^{\D}
             \end{pmatrix}^{T} $, we get 
             \begin{equation}
                    \begin{split}
                        S &= \begin{pmatrix}
                    w_{13}^{\Im} & w_{13}^{\Re} & 0 \\
                    -w_{13}^{\Re} & w_{13}^{\Im} & -2 w \\
                    -w w_{13}^{\Re} & w w_{13}^{\Im} & 2 |w_{13}|^{2} \\
                 \end{pmatrix}, \\[0.2cm]
                 \det S &= 2 |w_{13}|^{4} + 2 w^{2}  |w_{13}|^{2} = \\[0.2cm]
                 &= 2 |w_{13}|^{2} \left(|w_{13}|^{2} + w^{2} \right) \neq 0.
                    \end{split}
             \end{equation}
             Therefore, we obtain $ B_{13}^{\Re}, \, B_{13}^{\Im}, \, B_{3}^{\D} \in \mathcal{L}(i H, \, iV)$. Hence
             \begin{align}
                     &L_{8} := [B_{3}^{\D}, \, L_{1}] = v_{23}^{\Im} B_{23}^{\Re} + v_{23}^{\Re} B_{23}^{\Im} = i \begin{pmatrix}
                         0 & 0 & 0 \\
                         0 & 0 & v_{23} \\
                         0 & v_{23}^{*} & 0  
                     \end{pmatrix}, \label{CTLSCC_eq:technoname3} \\[0.2cm]
                     &L_{9} := \ad^{2}_{B_{3}^{\D}}(L_{1}) = -v_{23}^{\Re} B_{23}^{\Re} + v_{23}^{\Im} B_{23}^{\Im} = \begin{pmatrix}
                         0 & 0 & 0 \\
                         0 & 0 & v_{23} \\
                         0 & -v_{23}^{*} & 0 
                     \end{pmatrix}, 
             \end{align}
             and for $ L = \begin{pmatrix}
                 L_{8} & L_{9}
             \end{pmatrix}^{T} $ and $ B = \begin{pmatrix}
                 B_{23}^{\Re} & B_{23}^{\Im}
             \end{pmatrix}^{T} $
             \begin{equation}  \label{CTLSCC_eq:technoname4} 
                 S = \begin{pmatrix}
                     v_{23}^{\Im} & v_{23}^{\Re} \\
                     -v_{23}^{\Re} & v_{23}^{\Im}
                 \end{pmatrix}, \quad \det S = (v_{23}^{\Re})^{2} + (v_{23}^{\Im})^{2} = |v_{23}|^{2} \neq 0. 
             \end{equation}
             Thus, also $ B_{23}^{\Re}, \, B_{23}^{\Im} \in  \mathcal{L}(i H, \, iV) $, and the system is controllable by Theorem~\ref{CTLSCC_th:algebrastatement}.
            \vspace{0.1cm}
            \item $ |v_{12}| = |v_{23}| $ and $ v_{1} = v_{3} $: \label{CTLSCC_enum:1.2}
            \vspace{0.1cm}
            
            First, if $ v_{1} = v_{3} = v $, the system may be treated as having no controlled energy-level shifts. This follows because $ iV $ can be brought to the form given in~\eqref{CTLSCC_eq:technoname2} by subtracting the term
            \begin{equation}
                 v \, i \mathbb{I} + (v_{2} - v) \, i H  = i \operatorname{diag}(v, \, v_{2}, \, v).
            \end{equation}
            Denote the absolute values $ r := |v_{12}| = |v_{23}| $, $  r' := |v_{13}| $, $ \varphi_{ij} := \arg(v_{ij}) $ and consider $ v_{1} = v_{2} = v_{3} = 0 $. Then
            \begin{align}
                \tilde{L}_{1} &:= r^{-1} L_{1} = -\cos(\varphi_{12}) B_{12}^{\Re} + \sin(\varphi_{12}) B_{12}^{\Im} + \cos(\varphi_{23}) B_{23}^{\Re} - \sin(\varphi_{23}) B_{23}^{\Im} = \\[0.2cm]
                &= \begin{pmatrix}
                    0 & e^{i \varphi_{12}} & 0 \\
                    e^{-i (\varphi_{12} + \pi)} & 0 & e^{i \left(\varphi_{23} + \pi\right)} \\
                    0 & e^{-i \varphi_{23}} & 0 
                \end{pmatrix}, \nonumber \\[0.2cm]
                 \tilde{L}_{2} &:= r^{-1} L_{2} = \sin(\varphi_{12}) B_{12}^{\Re} + \cos(\varphi_{12}) B_{12}^{\Im} + \sin(\varphi_{23}) B_{23}^{\Re} + \cos(\varphi_{23}) B_{23}^{\Im} = \\[0.2cm]
                &= \begin{pmatrix}
                    0 & e^{i \left(\varphi_{12} + \frac{\pi}{2} \right)} & 0\\
                    e^{-i \left(\varphi_{12} - \frac{\pi}{2} \right)}& 0 & e^{i \left(\varphi_{23} + \frac{\pi}{2} \right)}  \\
                    0 & e^{-i \left(\varphi_{23} - \frac{\pi}{2} \right)} & 0
                \end{pmatrix}, \nonumber \\
                 \tilde{L}_{3} &:= (r')^{-1} L_{3} = \sin(\varphi_{13}) B_{13}^{\Re} + \cos(\varphi_{13}) B_{13}^{\Im}
                = \begin{pmatrix}
                    0 & 0 & e^{i \left(\varphi_{13} + \frac{\pi}{2} \right)}  \\
                    0 & 0 & 0 \\
                    e^{-i \left(\varphi_{13} - \frac{\pi}{2} \right)}  & 0 & 0 
                \end{pmatrix},
            \end{align}
            where $ L_{1}$, $L_{2}$, and $L_{3} $ are defined by the equations \eqref{CTLSCC_eq:L1}, \eqref{CTLSCC_eq:L2}, and \eqref{CTLSCC_eq:L3}, respectively.

            Next, 
            \begin{align}
                L_{10} &:= -[\tilde{L}_{1}, \, \tilde{L}_{3}] = \nonumber \\[0.2cm]
                &= \sin(\varphi_{13} - \varphi_{23}) B_{12}^{\Re} + \cos(\varphi_{13} - \varphi_{23}) B_{12}^{\Im} + \sin(\varphi_{13} - \varphi_{12}) B_{23}^{\Re} + \cos(\varphi_{13} - \varphi_{12}) B_{23}^{\Im} = \\[0.2cm]
                &= \begin{pmatrix}
                    0 & e^{i\left(\varphi_{13} - \varphi_{23} + \frac{\pi}{2} \right)} & 0 \\
                    e^{-i\left(\varphi_{13} - \varphi_{23} - \frac{\pi}{2} \right)} & 0 & e^{i\left(\varphi_{13} - \varphi_{12} + \frac{\pi}{2} \right)}  \\
                    0 & e^{-i\left(\varphi_{13} - \varphi_{12} - \frac{\pi}{2} \right)}  & 0
                \end{pmatrix}, \nonumber \\[0.2cm]
                L_{11} &:= -[\tilde{L}_{2}, \, \tilde{L}_{3}] = \nonumber \\[0.2cm]
                &= \cos(\varphi_{13} - \varphi_{23}) B_{12}^{\Re} - \sin(\varphi_{13} - \varphi_{23}) B_{12}^{\Im} - \cos(\varphi_{13} - \varphi_{12}) B_{23}^{\Re} + \sin(\varphi_{13} - \varphi_{12}) B_{23}^{\Im} = \\[0.2cm]
                &= \begin{pmatrix}
                    0 & e^{i\left(\varphi_{13} - \varphi_{23} + \pi \right)} & 0 \\
                    e^{-i\left(\varphi_{13} - \varphi_{23} \right)} & 0 & e^{i\left(\varphi_{13} - \varphi_{12} \right)}  \\
                    0 & e^{-i\left(\varphi_{13} - \varphi_{12} + \pi\right)}  & 0
                \end{pmatrix}. \nonumber 
            \end{align}
            For $ L = \begin{pmatrix}
                \tilde{L}_{1} & \tilde{L}_{2} & L_{10} & L_{11} 
            \end{pmatrix}^{T} $ and $ B = \begin{pmatrix}
                B_{12}^{\Re} & B_{12}^{\Im} & B_{23}^{\Re} & B_{23}^{\Im}
            \end{pmatrix}^{T} $ we have
            \begin{equation}
                \begin{split}
                    S &= \begin{pmatrix}
                        -\cos(\varphi_{12}) & \sin(\varphi_{12}) & \cos(\varphi_{23}) & -\sin(\varphi_{23}) \\
                        \sin(\varphi_{12}) & \cos(\varphi_{12}) & \sin(\varphi_{23}) & \cos(\varphi_{23}) \\
                        \sin(\varphi_{13} - \varphi_{23}) & \cos(\varphi_{13} - \varphi_{23}) & \sin(\varphi_{13} - \varphi_{12}) & \cos(\varphi_{13} - \varphi_{12}) \\
                        \cos(\varphi_{13} - \varphi_{23}) & -\sin(\varphi_{13} - \varphi_{23}) & -\cos(\varphi_{13} - \varphi_{12}) & \sin(\varphi_{13} - \varphi_{12})
                    \end{pmatrix}, \quad \\[0.2cm]
                    \det S &= -4 \sin^{2} \left(\varphi_{12} + \varphi_{23} - \varphi_{13} \right) \neq 0
                \end{split}
            \end{equation}
            under the assumption of the case \ref{CTLSCC_enum:1}. Thus, $ B_{12}^{\Re}, B_{12}^{\Im}, B_{23}^{\Re}, B_{23}^{\Im} \in \mathcal{L}(iH, iV), $ and the system is controllable.
        \end{enumerate}

        \vspace{0.1cm}
        \item $ \displaystyle  \frac{v_{13}}{v_{12} v_{23}} \in \mathbb{R} $ and $ \displaystyle \frac{v_{13}}{v_{12} v_{23}}  \left(|v_{23}|^{2} - |v_{12}|^{2} \right) \neq v_{3} - v_{1} $: \label{CTLSCC_enum:2}
        \vspace{0.1cm}
        
        In this case, we have
        \begin{equation}
            \begin{split}
                & \frac{v_{13}}{v_{12} v_{23}} L_{4} - L_{3} + \left(\frac{v_{13}}{v_{12} v_{23}} |v_{12}|^{2} - v_{1} \right)(i H - i \mathbb{I}) = \\[0.2cm]
                &= \left[\left(\frac{v_{13}}{v_{12} v_{23}} |v_{23}|^{2} - v_{3} \right) - \left(\frac{v_{13}}{v_{12} v_{23}} |v_{12}|^{2} - v_{1}\right)\right] B_{3}^{\D}.
            \end{split}
        \end{equation}
        This means that $ B_{3}^{\D} \in \mathcal{L}(iH, iV) $ under the assumption of the case \ref{CTLSCC_enum:2}, and 
        \begin{equation}
            B_{1}^{\D} = i \mathbb{I} - iH - B_{3}^{\D} \in \mathcal{L}(iH, iV).
        \end{equation}
         The remaining controllability argument repeats the argument in the equations \eqref{CTLSCC_eq:technoname3}--\eqref{CTLSCC_eq:technoname4} for $ B_{1}^{\D} $ ($  B_{12}^{\Re}, \, B_{12}^{\Im} \in  \mathcal{L}(i H, \, iV)  $) and $ B_{3}^{\D} $ ($ B_{23}^{\Re}, \, B_{23}^{\Im} \in  \mathcal{L}(i H, \, iV) $).  

        \item $ \displaystyle  \frac{v_{13}}{v_{12} v_{23}} \in \mathbb{R} $ and $ \displaystyle \frac{v_{13}}{v_{12} v_{23}}  \left(|v_{23}|^{2} - |v_{12}|^{2} \right) = v_{3} - v_{1} $: \label{CTLSCC_enum:3}
        \vspace{0.1cm}
        
        This case implies that
        \begin{equation}
            L_{3} =  \frac{v_{13}}{v_{12} v_{23}} L_{4} + c(i H - i \mathbb{I}),
        \end{equation}
        where 
        \begin{equation}
            c =  \frac{v_{13}}{v_{12} v_{23}}  |v_{12}|^2-v_1 =  \frac{v_{13}}{v_{12} v_{23}} |v_{23}|^2 - v_{3}.
        \end{equation}
        Therefore, the dynamical Lie algebra can be considered as the one generated by the elements $ i H, \ L_{1}, \, L_{2}, \, L_{4} $ (note that $ iV $ is excluded as a linear combination of $ L_{1}, \, L_{2} $, and $ L_{4} $  according to \eqref{CTLSCC_eq:L3}). But, as shown in Table~\ref{CTLSCC_tab:commutators}, this set is closed under the commutation operation.
            
    \begin{table}[h!]
        \centering
        \begin{tabular}{|c|c|c|c|c|c|}
            \hline
            \multicolumn{2}{|c|}{\multirow{2}{*}{$ \mathbf{[A, \, B]} $}} 
            & \multicolumn{4}{c|}{$\mathbf{B}$} \\ \cline{3-6}
            \multicolumn{2}{|c|}{} 
            & $ iH $ & $ L_{1} $ & $ L_{2} $ & \, $ L_{4} $ \\ \hline
            \multirow{4}{*}{$ \mathbf{A} $} & $ i H $ & 0 & $ -L_{2} $ & $ L_{1} $ & 0 \\ \cline{2-6}
            & $ L_{1} $ & $ L_{2} $ & 0 & $ 2 L_{4} - 2\left(|v_{12}|^{2} + |v_{23}|^{2} \right) i H $ & $ -\left(|v_{12}|^{2} + |v_{23}|^{2} \right) L_{2} $ \\ \cline{2-6}
            & $ L_{2} $ & $ -L_{1} $ & $ -2 L_{4} + 2\left(|v_{12}|^{2} + |v_{23}|^{2} \right) i H $ & 0 & $ \left(|v_{12}|^{2} + |v_{23}|^{2} \right) L_{1} $ \\ \cline{2-6}
            & $ L_{4} $ & 0 & $ \left(|v_{12}|^{2} + |v_{23}|^{2} \right) L_{2} $ & $ -\left(|v_{12}|^{2} + |v_{23}|^{2} \right) L_{1} $ & 0 \\ \hline
        \end{tabular}
        \caption{Commutators of elements from the set $ \left\{ iH, \, L_{1}, \, L_{2}, \, L_{4}\right\} $. Together with $i\mathbb I$, the linear span of these elements is closed under commutation. Hence their commutators do not generate new linearly independent elements so that $\dim\mathcal L(iH,iV)\le5$, which is insufficient for operator controllability.} \label{CTLSCC_tab:commutators}
    \end{table}

    This means that $ \dim(\mathcal{L}(iH, iV)) \leqslant 5 $, which is less than $ 9 $ (or $ 8 $ for the traceless case). Thus, the system is not controllable by the LARC. 
        
    \end{enumerate}
\end{clproof}
If it is known that the system has no controlled energy-level shifts, or, more generally, if it is known that it has no mutually distinct controlled energy-level shifts of degenerate levels, then the condition simplifies as follows.
\begin{corollary*}
    Let the system $ (H, V) $ be nonregular with three allowed transitions. Assume also that
    \begin{equation}
        v_{i} = v_{j},
    \end{equation}
    which means that the system does not have relatively nontrivial controlled energy-level shifts of the degenerate levels. Then, the system is controllable iff 
    \begin{equation}
        \arg(v_{ij}) \neq \arg(v_{ik}) + \arg(v_{kj}) + \pi m \ \text{ or } \ |v_{ik}| \neq |v_{kj}| \quad (m\in\mathbb Z).
    \end{equation}
\end{corollary*}

\subsection{Role of the assumptions in Theorem~\ref{CTLSCC_th:shifts} on adding controlled energy-level shifts}\label{CTLSCC_subsec:highlight}

Here we present two counterexamples of three-level systems showing that neither of the two assumptions of Theorem~\ref{CTLSCC_th:shifts} can be omitted in general. Recall that this theorem declares preservation of controllability when adding controlled energy-level shifts to a controllable system~if
\begin{enumerate}
    \item The new dynamical Lie algebra $ \mathcal{L}(iH, \, iV') $ contains the linear combination $ \displaystyle \sum\limits_{k \in K}  v_{k} B_{k}^{\D} $ of basis elements $ B_{k}^{\D} $ corresponding to the original controlled energy-level shifts;
    \item The system is $ V $--regular.
\end{enumerate}
One counterexample demonstrates the role of the Lie-algebra inclusion condition for the linear combination representing the original controlled energy-level shifts, while the other demonstrates the role of $ V $--regularity.

First, consider a family of regular but not strongly regular systems ($ E_{3} = E_{2} + \Delta{E} = E_{1}+2 \Delta E $) with the following interaction Hamiltonians $ V $ and $ V' $   
\begin{equation}
    V = \begin{pmatrix}
        0 & e^{i\varphi_{12}} & 0 \\
        e^{-i\varphi_{12}} & 0 & e^{i\varphi_{23}} \\
        0 & e^{-i\varphi_{23}} & v
    \end{pmatrix} \quad \to \quad V' = \begin{pmatrix}
        -v & e^{i\varphi_{12}} & 0 \\
        e^{-i\varphi_{12}} & 0 & e^{i\varphi_{23}} \\
        0 & e^{-i\varphi_{23}} & v
    \end{pmatrix},
\end{equation}
where $ \varphi_{ij} \in \mathbb{R}, \, v \in \mathbb{R} \setminus \left\{0\right\}$, and $ \to $ indicates that one interaction Hamiltonian can be constructed from the other by adding the corresponding controlled energy-level shift. Controllability Condition \ref{CTLSCC_th:twonotstrong2} shows that the system with interaction Hamiltonian $ V $ is controllable: although
\begin{equation}
     |v_{12}| = |v_{23}| = 1,
\end{equation}
we have
\begin{equation}
     (v_{2} - v_{1}) = 0 \neq v = (v_{3} - v_{2}).
\end{equation}
Since the system is regular, the only remaining assumption of Theorem~\ref{CTLSCC_th:shifts} is that the basis element $ B_{3}^{\D}$ associated with the original controlled energy-level shift belongs to $ \mathcal{L}(iH, iV') $. However, the same three-level condition shows that this system is not controllable because $ (v'_{2} - v'_{1}) = (v'_{3} - v'_{2}) = v $. Thus, this requirement is the only one that is not satisfied, and it destroys controllability. 

Next, consider a family of nonregular systems ($ E_{1} = E_{3} \neq E_{2} $) with
\begin{equation}
    \begin{split}
        & V = \begin{pmatrix}
        0 & r_{12} e^{i \varphi_{12}} & r_{12} r_{23} e^{i (\varphi_{12} + \varphi_{23})} \\
        r_{12} e^{-i \varphi_{12}} & 0 & r_{23} e^{i \varphi_{23}} \\
        r_{12} r_{23} e^{-i (\varphi_{12} + \varphi_{23})} & r_{23} e^{-i \varphi_{23}} & 0
    \end{pmatrix} \quad \to \\[0.2cm]
    \to \quad &V' = \begin{pmatrix}
        0 & r_{12} e^{i \varphi_{12}} & r_{12} r_{23} e^{i (\varphi_{12} + \varphi_{23})} \\
        r_{12} e^{-i \varphi_{12}} & 0 & r_{23} e^{i \varphi_{23}}\\
        r_{12} r_{23} e^{-i (\varphi_{12} + \varphi_{23})} & r_{23} e^{-i \varphi_{23}} & r_{23}^{2} - r_{12}^{2}
    \end{pmatrix},
    \end{split}
\end{equation}
where $ \varphi_{ij} \in \mathbb{R}$, $ r_{12}, r_{23} > 0 $, and $ r_{12} \neq r_{23} $. The initial system is controllable by Controllability Condition \ref{CTLSCC_th:threenotregular} as
\begin{equation}\label{CTLSCC_eq:technoname9}
    \frac{v_{13}}{v_{12} v_{23}} \cdot \left(|v_{23}|^2 - |v_{12}|^{2}\right) = r_{23}^{2} - r_{12}^{2} \neq 0 = v_{3} - v_{1},
\end{equation}
and it also does not have any controlled energy-level shifts. Therefore, $ V $--regularity is the only condition to be met for the system with interaction Hamiltonian $ V' $ to be controllable. However, it is not controllable since the left-hand side of \eqref{CTLSCC_eq:technoname9} remains the same for $ v'_{ij} $ while
\begin{equation}
    v'_{3} - v'_{1} = r_{23}^{2} - r_{12}^{2}.
\end{equation}
This shows that $ V $--regularity is also crucial for the considered sufficient condition.

\subsection{Complete classification of operator controllability for closed three-level quantum systems}\label{CTLSCC_subsec:finalclassification}
In this subsection, we combine the controllability criteria obtained for various three-level systems into a general complete classification, which is summarized for non-trivial three-level systems in Table \ref{CTLSCC_tab:result} and Figures~\ref{CTLSCC_fig:results_green},~\ref{CTLSCC_fig:results_yellow}.

\begin{figure}[h!]
\centering
    \includegraphics[width=\linewidth]{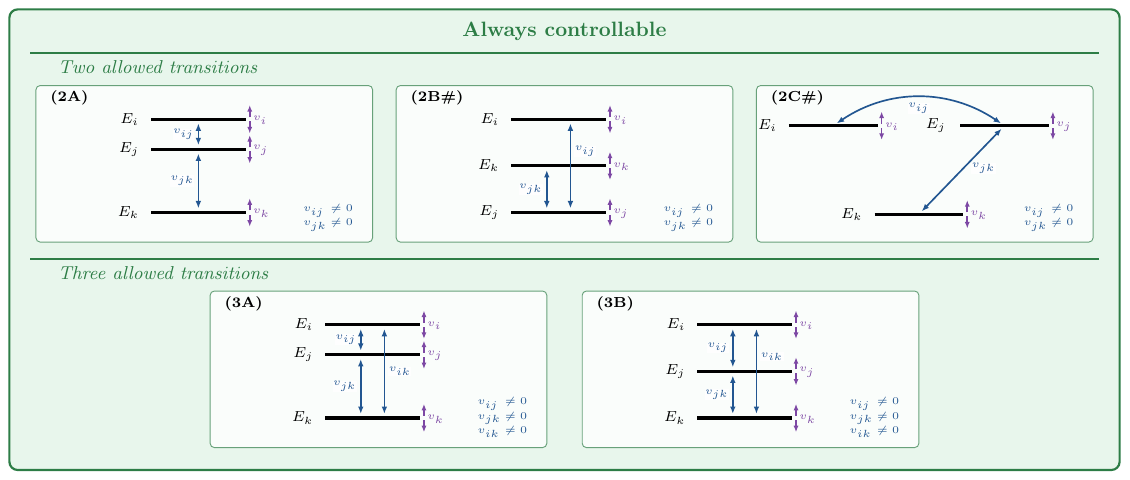}
    \caption{Classes of three-level systems that are operator controllable without additional restrictions on the interaction Hamiltonian matrix elements. The horizontal black lines represent energy levels, blue arrows indicate allowed population-transfer pathways, and short purple arrows indicate possible controlled energy-level shifts.}
    \vspace{0.2cm}
    \label{CTLSCC_fig:results_green}
\end{figure}

\begin{figure}[h!]
\centering
    \includegraphics[width=\linewidth]{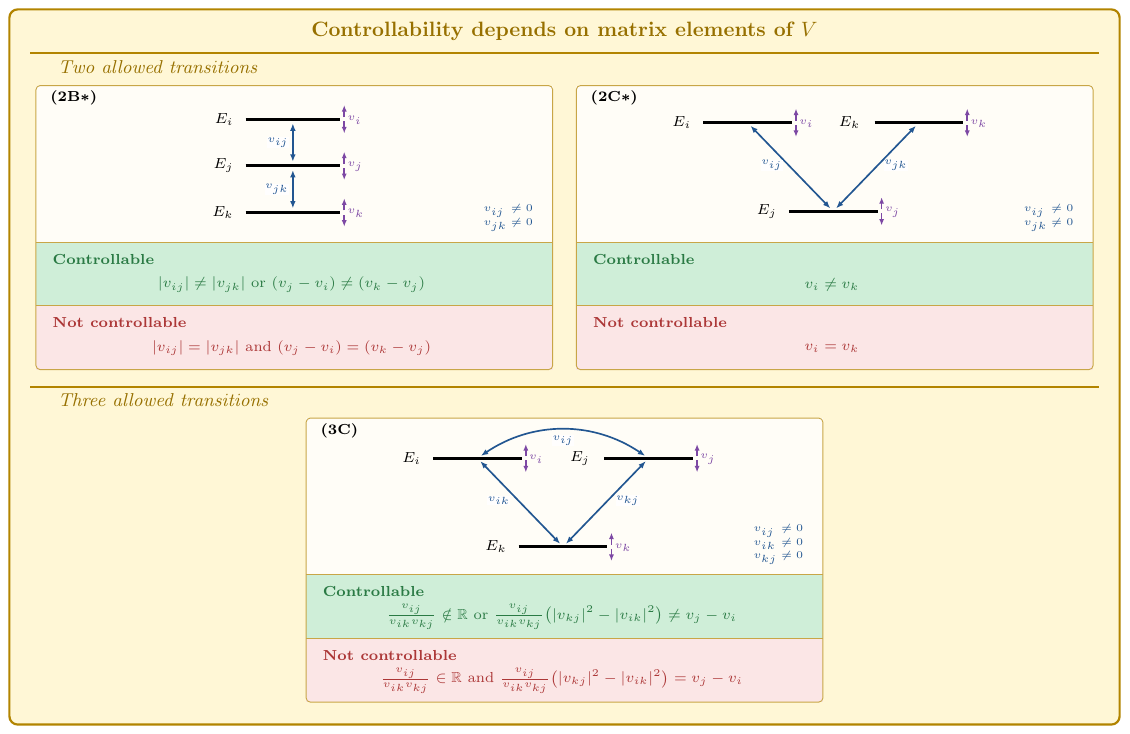}
    \caption{Classes of three-level systems whose operator controllability depends on the interaction Hamiltonian matrix elements. Each displayed system is operator controllable if and only if the condition shown in green box below its diagram is satisfied; otherwise, it is not controllable.}
    \label{CTLSCC_fig:results_yellow}
\end{figure}

\begin{table}[h!]
    \small
    \centering
    \begin{tabular}{|c|c|c|}
    \hline
    & & \\
    \multirow{2}{*}{\textbf{Three-level systems}} & \textbf{Two allowed transitions (2)} & \textbf{Three allowed transitions (3)} \\
    & $v_{ik} = 0, \,  v_{ij}, v_{jk} \neq 0 $ &  $ v_{ij}, \, v_{jk}, \, v_{ik} \neq 0 $ \\
    & & \\
    \hline
    & & \\
    \textbf{Strongly regular (A)} & \multirow{2}{*}{Controllable} & \multirow{2}{*}{Controllable} \\ 
    $ |\Delta E_{ij}| \neq |\Delta E_{jk}| \neq |\Delta E_{ik}| \neq |\Delta E_{ij}| $ & & \\
    & & \\
    \hline
    & & \\
    & Controllable for $ \Delta E_{ij} = \Delta E_{jk} $ \textbf{(\textasteriskcentered)}  & \\
    \textbf{Regular but not strongly reg. (B)} & iff $ |v_{ij}| \neq |v_{jk}| $ or  & \\
    $ \Delta E_{ij} = \Delta E_{jk} $ or $ \Delta E_{ik} = \Delta E_{kj} $, & $ (v_{j} - v_{i}) \neq (v_{k} - v_{j}) $ &  Controllable  \\
    and $ E_{i} \neq E_{j} \neq E_{k} \neq E_{i} $  &  &  \\
    &  Controllable for $ \Delta E_{ik} = \Delta E_{kj}$ \textbf{(\#)} & \\  
    & & \\
    \hline
    &  & \\
    \textbf{Nonregular (C)} & Controllable for $ E_{i} = E_{k} $ \textbf{(\textasteriskcentered)} & Controllable iff \\
    $ E_{i} = E_{j} $ or $ E_{i} = E_{k} $ for two trans. &  iff $ v_{i} \neq v_{k} $;  & $ \arg(v_{ij}) \neq \arg(v_{ik}) + \arg(v_{kj}) + \pi m $ or \\
    $ E_{i} = E_{j} $ for three trans. &  Controllable for $ E_{i} = E_{j} $ \textbf{(\#)} & $ (v_{ij}/v_{ik} v_{kj}) \, (|v_{kj}|^{2} - |v_{ik}|^{2}) \neq (v_{j} - v_{i}) $ \\
    & & \\
    \hline
    \end{tabular}
    \caption{Complete operator controllability classification of non-trivial three-level closed quantum systems summarizing Controllability Conditions \ref{CTLSCC_th:twostrong}--\ref{CTLSCC_th:threenotregular}. Systems with $H=c\mathbb I$ or with a disconnected structural graph $\mathcal G_V$ are not controllable and are not included in the table. For every remaining case, the displayed statement gives necessary and sufficient conditions for operator controllability. Indices $ i, j$, and $ k $ are pairwise distinct and label the system energy levels; $ m\in\mathbb Z $ in the criterion for nonregular systems with three allowed transitions. We use $v_{ji}=v_{ij}^{*}$ as a notational convention for the corresponding matrix elements. The symbols $*$ and $\#$ denote, respectively, the presence and absence of level-spacing-degenerate transitions among the two allowed transitions.
    \vspace{0.35cm}
    } \label{CTLSCC_tab:result}
\end{table}

Systems with $H=c\mathbb I$ or with a disconnected structural graph $\mathcal G_V$ are not controllable and are not included in the table. All strongly regular systems, as well as regular but not strongly regular systems with three allowed transitions, are completely controllable without any restrictions on the interaction Hamiltonian. Controllability of regular but not strongly regular systems with only two allowed transitions depends on their structure: they are controllable precisely if these transitions are not level-spacing degenerate or if they are degenerate and either $|v_{ij}|\ne|v_{jk}|$  or $(v_j-v_i)\ne(v_k-v_j)$. Nonregular systems with two allowed transitions without level-spacing-degenerate transitions are controllable, whereas nonregular systems with such transitions are controllable precisely when $v_i\ne v_k$, where $E_i=E_k$. Finally, nonregular systems with three allowed transitions are controllable precisely when either $\arg(v_{ij})\ne\arg(v_{ik})+\arg(v_{kj})+\pi m$ for $m\in\mathbb Z$, or
\[
\frac{v_{ij}}{v_{ik}v_{kj}}\left(|v_{kj}|^2-|v_{ik}|^2\right)\ne v_j-v_i.
\]Together with the immediate noncontrollability of systems with $H=c\mathbb I$ or a disconnected structural graph $\mathcal G_V$, these statements constitute the complete operator-controllability classification of closed three-level systems driven by a single coherent control. This classification covers all combinations of energy-level degeneracies, level-spacing degeneracies, two or three allowed transitions, as well as the presence or absence of controlled energy-level shifts, and gives necessary and sufficient conditions for operator controllability in each case.

\section{Conclusions}\label{CTLSCC_sec:conclusion}
This article investigates operator controllability of finite-dimensional closed quantum systems driven by a single coherent control. Previous work has established general Lie-algebraic controllability criteria and has studied various specific classes of three-level systems. Building on these results, we derive a complete classification of operator controllability for three-level systems in terms of the matrix elements of their free and interaction Hamiltonians and the structural features encoded by these elements.

We present Theorem~\ref{CTLSCC_th:shifts}, which provides a sufficient condition for preserving controllability when controlled energy-level shifts are added to an initially controllable system. We also establish Theorem~\ref{CTLSCC_th:algebrastatement}, which shows the inclusion of $ \mathfrak{su}(N) $ in the dynamical Lie algebra using a graph constructed with special skew-Hermitian basis elements belonging to that algebra.

Using these theorems together with established general criteria, Theorem~\ref{CTLSCC_th:two_transitions_classification} with Controllability Conditions~\ref{CTLSCC_th:twostrong}--\ref{CTLSCC_th:twonotregular2} extends previous results for systems with two allowed energy-level transitions to include controlled energy-level shifts. Theorem~\ref{CTLSCC_th:three_transitions_classification} with Controllability Conditions~\ref{CTLSCC_th:threestrong}, \ref{CTLSCC_th:threenotstrong}, and~\ref{CTLSCC_th:threenotregular} provides the corresponding complete analysis for systems with three allowed transitions. The resulting classification is summarized in Table~\ref{CTLSCC_tab:result} and Figures \ref{CTLSCC_fig:results_green}, \ref{CTLSCC_fig:results_yellow}.

\section*{Acknowledgments}
This work was supported by the Russian Science Foundation under grant no. 22-11-00330-P, \url{https://rscf.ru/project/22-11-00330/}.

\section*{Data availability statement}
All data that support the findings of this study are included within the article.

\end{document}